\documentclass[a4paper,USenglish, autoref, thm-restate]{lipics-v2021}

\pdfoutput=1 

\hideLIPIcs  

\title{Keeping it sparse: Computing Persistent Homology revisited}

\titlerunning{Keeping it sparse} 

\author{Ulrich Bauer}{Department of Mathematics, TUM School of CIT, Technical University of Munich, Germany, Germany}{ulrich.bauer@tum.de}{https://orcid.org/0000-0002-9683-0724}{DFG Collaborative Research Center SFB/TRR 109 ``Discretization in Geometry and Dynamics''}

\author{Talha Bin Masood}{Department of Science and Technology, Linköping University, Sweden}{talha.bin.masood@liu.se}
{https://orcid.org/0000-0001-5352-1086}{Swedish Research Council (VR) grant number 2023-04806}

\author{Barbara Giunti}{Department of Mathematics, Graz University of Technology, Austria}{bgiunti@tugraz.at}{https://orcid.org/0000-0002-3500-8286}{Austrian Science Fund (FWF) grant number P 29984-N35 and P 33765-N}

\author{Guillaume Houry}{{\'E}cole Polytechnique, France}{guillaume.houry@live.fr}{https://orcid.org/0000-0002-1918-0545}{}

\author{Michael Kerber}{Department of Mathematics, Graz University of Technology, Austria}{kerber@tugraz.at}{https://orcid.org/0000-0002-8030-9299}{Austrian Science Fund (FWF) grant number P 29984-N35 and P 33765-N}

\author{Abhishek Rathod}{Computer Science Department, Ben Gurion University, Israel}{arathod@post.bgu.ac.il}{https://orcid.org/0000-0003-2533-3699}{DFG Collaborative Research Center SFB/TRR 109 ``Discretization in Geometry and Dynamics''}

\authorrunning{Bauer, Bin Masood, Giunti, Houry, Kerber, Rathod} 

\Copyright{CC-BY} 



\relatedversion{} 

\acknowledgements{}

\nolinenumbers 

\usepackage[algo2e,vlined,ruled,linesnumbered]{algorithm2e} 
\SetKwIF{If}{ElseIf}{Else}{if}{}{else if}{else}{end if}%
\SetKwFor{While}{while}{}{end while}%
\SetKwFor{For}{for}{}{end for}%
\DontPrintSemicolon
\usepackage{algorithm}
\usepackage{algpseudocode}

\usepackage{booktabs} 

\usepackage{hyperref}
\usepackage[noabbrev,nameinlink,capitalise]{cleveref}

\newcommand{\define}[1]{{\bf \boldmath{#1}}}

\usepackage{tikz}
\usetikzlibrary{arrows,arrows.meta}
\usetikzlibrary{backgrounds}
\pgfdeclarelayer{foreground}
\pgfsetlayers{background,main,foreground}

\tikzset{
dot/.style = {circle, fill, minimum size=#1,
              inner sep=0pt, outer sep=0pt},
}

\usepackage{adjustbox}
\usepackage{subcaption}

\newcommand{\betti}[1]{\beta_{#1}}
\newcommand{\altcomplex}{\mathsf{K}}
\newcommand{\complex}{\mathsf{K}}

\newcommand{\lowp}[1]{\mathrm{piv}\left(#1\right)}
\newcommand{\sizep}[1]{\# #1}

\newcommand{\phat}{PHAT\xspace}

\newcommand{\ignore}[1]{}

\newcommand{\persbetti}[1]{\beta_{#1}}

\DeclareMathOperator*{\rank}{rank}

\usepackage{multirow}

\begin{document}

\maketitle

\begin{abstract}
In this work, we study several variants of matrix reduction via Gaussian elimination 
that try to keep the reduced matrix sparse.
The motivation comes from the growing field of topological data analysis
where matrix reduction is the major subroutine to compute barcodes, the main invariant therein.
We propose two novel variants of the standard algorithm, called swap and retrospective reductions. 
We test them on a large collection of data against other known variants to compare their efficiency, and we find that sometimes they provide a considerable speed-up.
We also present novel output-sensitive bounds for the retrospective variant which better explain the discrepancy between the cubic worst-case complexity bound and the almost linear practical behavior of matrix reduction. 
Finally, we provide several constructions on which one of the variants performs strictly better than the others.
\end{abstract}

\section{Introduction}
\subsection{Motivation}
Persistent homology is arguably the most important tool in the thriving area of topological data analysis. 
The presence of efficient algorithms for computing the \define{barcode}, its main invariant, has been an important contributing factor to its success.
In the so-called \define{standard algorithm}~\cite{top_per_and_simp}, this computation boils down to a ``restricted'' Gaussian elimination of a boundary matrix of a filtered simplicial complex: 
no swapping of rows or columns, only column operations are allowed, and the additions are only from left to right. 

Since Gaussian elimination has cubic worst-case complexity in the size of the matrix, barcodes can be computed in polynomial time. 
This makes it already efficient compared to many other topological invariants, which are usually (NP-)hard to compute or not computable at all.
In practice, however, the performance is even better: 
we observe a close-to-linear practical behavior that allows computing the barcode for matrices with billions of columns~\cite{phat_paper,roadmap}.
This happens because the matrices to be reduced are \define{sparse} (that is, the number of nonzero entries per column is a small constant) and tend to remain so during the reduction.

Gaussian elimination on sparse matrices is cheaper because column additions can be performed more efficiently with appropriate choices of sparse matrix data structures~\cite{phat_paper}.
However, the sparsity of the input matrix does not alter the worst-case cubic bounds for matrix reduction. 
Indeed, carefully crafted constructions force a matrix to become dense in the elimination process, making later column additions expensive~\cite{morozov2005worst-case}. 
On the other hand, such situations seem to be pathological and do not happen in practice (and also not on average~\cite{average_complexity}).
Therefore, the practical efficiency of reduction procedures can be linked to the preservation of the sparsity of the matrix during the elimination process.

The standard algorithm has been optimized in several ways, exploiting the special structure of boundary matrices.
This led to significant further improvements in practice (which are discussed below). 
However, the impact of sparsity on the reduction has not been adequately studied.
We pose the following questions in this paper:
are there variants of the standard algorithm, possibly performing different operations, that keep the matrix sparser than the standard one, and do these variants exhibit better practical behavior than existing methods?

\subsection{Results}
We first observe that some of the restrictions imposed on the Gaussian elimination can be relaxed: any operation preserving the ranks of certain submatrices is allowed (\cref{cor:C_legit}). 
This observation enables us, for example, to swap certain columns and to perform some right-to-left column operations. 

With this insight, we then introduce two new variants of the standard algorithm. 
The first one, called the \define{swap algorithm}, introduces one extra rule:
before adding a column $c_1$ to a column $c_2$, it swaps $c_1$ and $c_2$ if  $c_2$ is sparser (i.e., has fewer nonzero entries).
Note that checking the size of a column, as well as swapping two columns, requires constant time for most matrix representation types, and hence the overhead of this variant is negligible as long as the column data structure allows for easy retrieval of the size. 
We show in extensive experimental tests that the swap algorithm is usually competitive with the fastest known algorithms and sometimes leads to significant speedup.

The other variant is called the \define{retrospective algorithm}.
It is based on the (well-known) idea that, once the \textbf{pivot} of a column has been found, we can perform additional column operations to further eliminate entries in the column (sometimes, this is called the ``full'' or the ``exhaustive'' reduction).  
The retrospective variant pushes this idea further: it eliminates entries also via right-to-left additions of newly reduced columns.
We show significant speedup over the state-of-the-art by experimental comparison for this variant as well.
Moreover, the retrospective strategy links the non-zero entries of a column with the (persistent) homology classes at that step, providing complexity bounds that depend on the topology of the underlying data set and are therefore output-sensitive.

For practical efficiency, both variants are combined with existing improvements of the standard algorithm, namely the \define{clear} and \define{compress optimization}~\cite{clearcompress} which typically save a lot of operations in matrix reduction.

We also show that none of the proposed algorithms is strictly better than the others
in the following sense. 
Having chosen one of the three algorithms (standard, swap, and retrospective reduction), we can find a family of inputs for which the chosen algorithm performs a linear number of operations whereas the other two have quadratic complexity. 

We also investigate other variants to ensure sparsity. 
For instance, during the exhaustive reduction, we generate several intermediate columns which are all valid representations for the rest of the algorithm, and we pick the sparsest column among these.
Remarkably, whilst this strategy appears to improve on both the standard and the exhaustive variants, in practice it performs worse than both of them. 
This shows that ensuring sparsity is not the only reason for the good practical performance of computing barcodes.

\subsection{Related work}
The \phat library~\cite{phat_paper} contains a collection of algorithms and matrix representation types to test various approaches for Gaussian elimination on boundary matrices in a unified framework.
Our work contributes several new algorithms to the \phat library.
We confirm the earlier observation that the quality of different elimination strategies significantly depends on the chosen data structure.

There are numerous other libraries to compute the barcode; we just refer to some comparative studies~\cite{minimal_userguide,roadmap,Cufar_Virk}. 
We point out that, in there, all tested libraries include further functionalities, in particular, generating a boundary matrix out of a point cloud, whereas \phat, as well as the present paper, focuses entirely on the Gaussian elimination step. 
\phat is among the most efficient libraries for this substep,
as demonstrated, for instance, in~\cite{phat_paper} and \cite{roadmap}.

Even more efficient algorithms have been developed for special cases of simplicial complexes, for instance, Vietoris--Rips complexes~\cite{bauer2021ripser, henselman2016matroid} and cubical complexes~\cite{gvt-dms,kaji2020cubical,wagner-socg-2023, Wagner2012}. 
The improvements in these specialized implementations are not based on optimized reduction strategies such as the ones considered in the present paper. However, not all of the reduction strategies considered here are compatible with these implementations, which is why we focus on the general-purpose framework PHAT for our comparisons.
There have also been several approaches that have focused specifically on parallel and distributed computation~\cite{clearcompress,dipha,lewis2015parallel,zhang2020gpu} for performance gains.
Another approach, taken for example in \cite{collapse1,collapse2}, is to simplify the filtration without modifying the homology. 
This approach effectively downsizes the matrix but does not sparsify it.

The best worst-case complexity for computing
the barcode is $O(n^\omega)$, where $\omega$ is the matrix multiplication constant~\cite{milos}. 
However, this approach is not based on Gaussian elimination and is not competitive in practice. 
There is no sub-cubic complexity bound known for any barcode algorithm based on Gaussian elimination.
The output-sensitive bounds that we derive still lead to cubic worst-case bounds, but can be tighter depending on the topological properties of the input. 
These bounds refine the bound by Chen and Kerber~\cite{twist}.

Computing homology with respect to $\mathbb{Z}$ coefficients requires computing the Smith normal form using integer Gaussian elimination~\cite[Chapter 1]{munkres}. 
In that context, fill-in of matrices is a minor concern, whereas a significant challenge comes from the fact that the size of intermediate entries during matrix reduction can grow exponentially~\cite{worstsmith}.

\section{Preliminaries}\label{sec_preliminaries}

\subsection{Matrix reduction} 
The algorithms presented in \Cref{variants} can be stated and implemented for other field coefficients with minor changes. 
However, for simplicity, throughout the paper, we work with $\mathbb{Z}_2$ coefficients.
Given a matrix $M$ with entries in $\mathbb{Z}_2$, $M^i$ denotes its $i$-th row, $M_j$ its $j$-th column, $M_j^i$ its element in position $i,j$, and $\sizep{M_j}$ the number of nonzero entries in $M_j$.
$N$ denotes the number of columns of $M$.

The \define{pivot} of a column, denoted by $\lowp{M_j}$, is the (row) index of the lowest nonzero element in $M_j$. 
A \define{left-to-right column operation} is the addition of $M_j$ to $M_i$ with $j<i$. 
A matrix is \define{reduced} if all its nonzero columns have pairwise distinct pivots.
The process of obtaining a reduced matrix using column additions is called \define{matrix reduction}. 
A \define{pivot pair} is a pair of indices $(i,j)$ such that $i=\lowp{M_j}$ in the reduced matrix.
\cref{algorithm_std_red} reduces the columns from left to right in order and is usually referred to as the \define{standard algorithm} for matrix reduction. 

\begin{algorithm2e}
\caption{Standard reduction}
\label{algorithm_std_red}
\KwIn{Boundary matrix $D$}
\KwOut{Reduced matrix $R$}
$R = D$ \;
\For{$j = 1,\dots,N$}{ 
	\While{$\lowp{R_{j'}}=\lowp{R_j}\neq 0$ for $j' < j$}{
		add $R_{j'}$ to $R_j$}
	} 
\end{algorithm2e}

\subsection{Filtered simplicial complex and boundary matrix} 
We apply matrix reduction on a class of matrices that arise from computational topology.

A \define{simplicial complex} $\altcomplex$ over a finite set $V$ is a collection of subsets (called \define{simplices}) of $V$ closed under inclusion, i.e. with the property that if $\sigma\in \altcomplex$ and $\tau\subset\sigma$, also $\tau\in \altcomplex$. 
A simplex with $(k+1)$-elements is called \define{$k$-simplex} and its dimension, $\dim(\sigma)$, is $k$. 
The dimension of $\complex$ is the maximal dimension of its simplices.
For $k=0,1,2$ the terms \define{vertices}, \define{edges}, and \define{triangles} are also used, respectively.
For a $k$-simplex $\sigma\in \altcomplex$, we call a $(k-1)$-simplex $\tau$ with $\tau\subset\sigma$ a \define{facet} of $\sigma$.
The set of facets of $\sigma$ is called its \define{boundary}. 

A \define{simplexwise filtered simplicial complex} is a sequence of nested simplicial complexes $\emptyset=\altcomplex_{0}\subseteq \cdots\subseteq \altcomplex_N=\altcomplex$ such that $\altcomplex_i\setminus \altcomplex_{i-1}= \{\sigma_{i}\}$ for all $i=1,\dots, N$. 
We denote the dimension of a simplex $\sigma_i$ by $d_i$.
The \define{boundary matrix} $D$ of a filtered simplicial complex is the $(N\times N)$-matrix such that $D_j^i=1$ if $\sigma_i$ is a facet of $\sigma_j$, and $0$ otherwise. 
In other words, the $j$-th column of $D$ encodes the boundary of the $j$-th simplex of the filtration.
Note that the boundary of a $k$-simplex consists of exactly $k+1$ facets, so under the reasonable assumption that the maximal dimension of a simplicial complex is a small constant, $D$ has only a constant number of nonzero entries in each column. 

\subsection{Persistence pairs}
Matrix reduction on boundary matrices reveals topological properties of the underlying filtered simplicial complex. 
We use standard notations for the necessary concepts that originate from (persistent) homology theory.
We also informally describe their topological meaning, although no deeper understanding of these concepts is required for the results of the paper.

Fixing a filtration boundary matrix $D$, matrix reduction yields a collection of pivot pairs $(i,j)$. 
The corresponding pair of simplices $(\sigma_i,\sigma_j)$ is called a \define{persistence pair}. 
For a persistence pair, $\dim(\sigma_j)=\dim(\sigma_i)+1$.
Informally, the meaning of a persistence pair is that when $\sigma_i$ is added to the filtered simplicial complex (at step $i$), it gives rise to a new ``hole'' in the complex (more precisely, a homology class). 
This hole disappears when $\sigma_j$ enters the filtered simplicial complex (e.g., $\sigma_j$ fills up that hole).
For formal definitions of these concepts, see~\cite[Sec. VII]{edel_harer2010computational}.

Pivot pairs of boundary matrices have special properties that are not true for other types of matrices: 
first of all, every pivot pair $(i,j)$ satisfies $i<j$, because a filtration boundary matrix is necessarily upper-triangular, and this property is preserved by matrix reduction.
Moreover, every index $j$ appears in at most one pivot pair:
this is based on the fact that inserting a $k$-simplex into a simplicial complex either creates a homology class in dimension $k$ or kills a homology class in dimension $k-1$~\cite[Pag. 154, see also Sec. V.4]{edel_harer2010computational}.
This allows us to classify simplices of the filtered simplicial complex into three types:
we call a simplex \define{positive} if it appears as the first entry in a persistence pair, \define{negative} if it appears as the second entry in a persistence pair, and \define{essential} if it does not appear in any persistence pair. In topological terms, essential simplices create a hole that is not filled up during the course of the filtration.

In what follows, we blur the difference between pivot pairs (of indices) and persistence pairs (of simplices) and identify $\sigma_i$, its index $i$ in the filtration, and the $i$-th column/row of the filtration boundary matrix. 
Hence, whenever convenient, we also talk about positive/negative indices and rows/columns.

\subsection{Clear and compress}
The special structure of boundary matrices allows for simple but effective speedups of matrix reduction. 
We describe two speedup heuristics that are relevant in this work, discussed extensively in~\cite{clearcompress}. 
Both are based on the observation that every index appears in at most one pivot pair. 

For the first heuristic, let us fix a negative index
$i$ and a column $D_j$ with $D_j^i\neq 0$.
It is then easy to see that $i$ cannot become the pivot of $D_j$ during the reduction process (because then, either $D_j$ itself or another column must end up with $i$ as the pivot, contradicting the assumption that $i$ is negative). 
Hence, we can simply remove the index $i$ from $D_j$ without changing the pivot pairs.
We call the process of removing all negative row indices from a column \define{compressing a column}.
For the second heuristic, let us fix a positive index $i$ and consider its column $D_i$. 
It can be readily observed that in the reduced matrix, $D_i$ cannot have a pivot because that would imply that $i$ is negative.
Therefore, $D_i$ can just be set to zero without changing the pivot pairs. 
We call this step \define{clearing a column}.

Note that to make use of clearing, the simplices of the simplicial complex have to be processed in decreasing dimensions; 
we refer to this variant of matrix reduction with clearing as \define{twist reduction}; the pseudocode can be obtained from~\cref{algorithm_swap_red} by removing \cref{swap_clear,swap_ifpivot}.
On the other hand, using compression requires proceeding in increasing dimensions, so clear and compress mutually exclude each other,  except for more sophisticated approaches~\cite{clearcompress}.
As shown in~\cite{phat_paper}, the twist reduction has a very satisfying practical performance and is the default choice in the \phat library.

\subsection{Column representations}
A crucial design choice when implementing matrix reduction is how to store the columns of the matrix. 
Since boundary matrices are initially sparse, and usually do not fill up too much in the reduction process, a dense vector over $\mathbb{Z}_2$ is a bad choice for its memory consumption. 
A data structure whose size is proportional to the number of nonzero entries in a column is preferred. 
A natural choice is to simply store the indices of nonzero entries in a sorted dynamic array (vector). 
Adding two such columns requires a merge of the two arrays canceling double-occurrences and is therefore proportional to the combined size of both columns. 
Storing the indices in heap order instead, we can realize the addition of $M_i$ to $M_j$ by inserting every entry of $M_i$ into the heap of $M_j$ which is possible in logarithmic time per entry in $M_i$. 
This approach does not eliminate double-occurrences of entries, but their removal can be delayed to a later point when either a pivot is queried or sufficiently many operations have been performed on a column such that a linear scan of the heap is affordable. Using such a ``lazy-heap'', the amortized cost of adding $M_j$ to $M_i$ is proportional to the size of $M_i$ plus logarithmic overhead.
We remark that sorted linked lists and balanced binary search trees can be used instead of vectors and heaps, respectively, to obtain similar performance guarantees in theory, but these data structures suffer from the lack of locality in storing the data which results in many cache misses and makes them significantly slower in most application scenarios.

Furthermore, while storing the non-zero indices of a column as a vector (or heap) is an efficient choice, it can be worthwhile to transform a column to a different data structure when reducing it. 
For instance, while storing every column as $0$-$1$-vector is prohibitive because of the memory consumption, it can be beneficial to ``expand'' the dense array of non-zero indices to a long array of $0$, and $1$, perform column additions on this column, and transform it back into a dense column once the column is reduced. 
An alternative is to use a lazy heap solely for the column to be reduced, and sorted arrays for all other columns.

The software library \phat~\cite{phat_paper} implements the all aforementioned data representations and several additional ones~-- we refer to the paper
for more extensive explanations of the data structures. 
We emphasize that the performance of matrix reduction depends not only on the reduction algorithm but on combining that algorithm with a suitable column representation; see~\cite[Tables~1 and 4]{phat_paper}.

\subsection{Dualization}
Given a simplex $\sigma$, the collection of all the simplices that have $\sigma$ as a facet is the \define{coboundary} of $\sigma$. 
As we did for the boundary, we can define the filtration \textbf{coboundary matrix}. 
The crucial observation is that the two matrices are (almost) anti-transposes of each other, and their pivot pairs are in bijection. 
We do not enter the details here, referring to~\cite{dualities} for the precise statements, but this observation has important consequences in practice.
As already observed~\cite{phat_paper, roadmap}, it is much faster to reduce the coboundary matrix than the boundary matrix for some inputs, in particular, for Vietoris--Rips filtrations.
An explanation of this phenomenon is given in~\cite{bauer2021ripser}.
Anti-transposing the matrix is called the \define{dualization process}, and it adds another degree of freedom when comparing the efficiency of the reductions. 

\section{Sparsification variants}\label{variants}
The Pairing Lemma~\cite[Pag. 154]{edel_harer2010computational} shows that the presence of a pivot pair $(i,j)$ is related to an inclusion-exclusion formula of ranks of $D[\geq \ast, \leq \bullet]$, submatrices of $D$ given by the last $\ast$ rows of the first $\bullet$ columns. 
It is usually used to prove the correctness of \cref{algorithm_std_red}, but it is much more general, as it implies:

\begin{corollary}\label{cor:C_legit}
Any matrix reduction algorithm that preserves the ranks of the submatrices $D[\geq i, \leq j]$, for all $i,j\in \{1,\dots, N\}$ is a valid barcode algorithm.
\end{corollary}

In other words, any reduction whose operations do not alter the ranks of lower-left submatrices will result in the same barcode decomposition.

\begin{proof}
Consider a reduction algorithm satisfying the hypothesis, and assume it obtains the pivot pair $(i,j)$.
By the Pairing Lemma, $(i,j)$ is a pivot pair if and only if 
$\rank\left(D[\geq i,\leq j]\right)-\rank\left(D[\geq i+1,\leq j]\right)+\rank\left(D[\geq i+1,\leq j-1]\right)-\rank\left(D[\geq i,\leq j-1]\right)=1$. 
Since all these ranks are preserved by the reduction, the claim follows.
\end{proof}

Notably, this interpretation of matrix reduction generalizes the common assumption that the reduced matrix $R$ is obtained from the original boundary matrix $D$ by left-to-right column additions, or equivalently, by multiplication with an invertible rank upper-triangular matrix.
While this restriction ensures that the reduction data determines a decomposition of the filtered chain complex (see, e.g., \cite{bauer2021ripser,edelsbrunner2017persistent}), the above observation shows that a weaker condition is sufficient if one is only interested in the barcode itself and not in the representative cycles or cocycles.
This insight opens the possibility of many new variants of the barcode algorithm that go beyond the use of left-to-right column additions. 
We now present some that try to keep the matrix sparse during the reduction.

\subsection{Swap reduction}
Our first major variant is based on the following simple idea: 
assume that the standard algorithm adds column $R_i$ to $R_j$ (hence $i<j$ and $R_i$ and $R_j$ have the same pivot). 
Before doing so, we can check first whether $R_j$ has fewer entries than $R_i$; in this case, we swap columns $R_i$ and $R_j$ first and perform the addition afterward (which still results in replacing $R_j$ with $R_i+R_j$). 
This swap is not only profitable in the column additions performed in this step, but also in every later column addition that involves column $R_i$. 

We call this variant the \define{swap reduction}. 
This variant appeared in Schreiber's PhD Thesis~\cite[p. 77]{SchreiberThesis} as a tool to control the size of the boundary matrix in a theorem on average complexity of matrix reduction;~see also~\cite{ks-casc}. 
The first indications of its practical performance appeared in~\cite{schmid_bachelor}.
We also point out that the swap reduction can easily be combined with the clearing optimization; 
see \cref{algorithm_swap_red} for the pseudocode for this variant.

For correctness, assume that columns $1,\ldots,i-1$ are already reduced and that column $i$ has the same pivot of a previous column, call it $i'$. 
Then the ranks of all lower-left submatrices up to $i$ are unchanged if we swap $i$ and $i'$. 
The correctness follows from \cref{cor:C_legit}.

\begin{algorithm2e}
	\caption{Swap reduction}
	\label{algorithm_swap_red}
	\KwIn{Boundary matrix $D$ of a simplicial complex of dimension $d$}
	\KwOut{Reduced matrix $R$}
	$R = D$ \;
        \For{$\delta=d,d-1,\ldots,0$}{
	  \For{$j = 1,\dots,N$}{ 
            \If{$R_j$ has simplex-dimension $\delta$} {
		\While{$\lowp{R_{j'}}=\lowp{R_j}\neq 0$ for $j' < j$}{
		\If{$\sizep{R_j}<\sizep{R_{j'}}$\label{swap_ifpivot}}{
			swap $R_j$ and $R_{j'}$\label{swap_clear}}
				add $R_{j'}$ to $R_j$}
		\If{$R_j\neq0$}{
		Set $R_i$ to $0$ for $i=\lowp{R_j}$}
		}
              }
            }
\end{algorithm2e}

\subsection{Exhaustive reduction}

We review the \define{exhaustive reduction}, discussed in~\cite{Edel-Olsb}, even if the idea was already present in~\cite{linkingnumbers,CZ05}. 
The idea is that after the pivot of the reduced matrix has been identified, further (left-to-right) column additions
are performed to eliminate nonzero entries with indices smaller than the pivot. 
Note that this algorithm produces the lexicographically smallest possible representative for the column given by left-to-right column additions. 
We omit the pseudocode for brevity (see \cite{Edel-Olsb}). 
The exhaustive reduction is combined with the compress-optimization (i.e. removing negative entries from a column before processing it)~\cite{CZ05}. 
In this way, the exhaustive reduction guarantees that the number of nonzero entries in $R_\ell$ after reduction is at most the number of homological classes in $\altcomplex_\ell$. 

\subsection{Retrospective reduction}
Our second major variant is the \define{retrospective reduction}, based on the idea of using (previous and subsequent) pivots to eliminate entries in a column when it needs to be added. 
An entry in $R_k^i$ is \define{unpaired at $\ell$} if there does not exist a pivot pair $(i,j)$ with $j \leq \ell$, and \define{paired at $\ell$} otherwise. 
Whenever we add a column $R_\ell$ to $R_k$, we first update $R_\ell$ by removing through appropriate column additions all entries that have been paired meanwhile. 
Note that, if the addition of $R_m$ to $R_\ell$ is needed for this purpose, then $R_m$ has to be updated first, so the step is recursive. 
The recursion stops because the pivot of a column is strictly decreasing in every recursive call.
Since these right-to-left column additions involve only entries whose index is smaller than the pivot in the respective column, none of them changes the rank of any lower-left submatrix. 
Therefore, this reduction is correct by \cref{cor:C_legit}.

The retrospective algorithm has the property that whenever $R_j$ gets added to another column during iteration $k$, its size is at most the number of homological classes persisting from $j$ to $k$ (see \cref{lem:backregular}). 
That is, it tries to sparsify columns ``that matter'', i.e. the columns that get added to other columns.

\begin{algorithm2e}
\caption{Retrospective reduction}
\label{algorithm_lr_red}
\KwIn{Boundary matrix $D$}
\KwOut{Reduced matrix $R$}
\SetKwFunction{main}{Main}
\SetKwFunction{red}{Reduce}
\SetKwProg{myproc}{Procedure}{}{}
\myproc{\main{$D$}}{
	$R = D$, $P=\emptyset$ \;
\For{$j= 1,\dots, N$ }{
Remove the negative entries from $R_j$ \;
\red{$j$}\label{lst:line:invoke1}
}
}
\myproc{\red{$j$}}{
\While{$\exists$ paired entries in $R_j$}{
    Let $\ell$ be the largest index for which $R_{j}^{\ell}$ is paired \;
    Add $\red{$P[\ell]$} $ to $R_j$
} 
\If{ $R_j \neq 0$}
{$P[ \textsc{Pivot}(R_j)] \gets j$}
\Return $R_j$
}
\end{algorithm2e}

\subsection{Representative cycles}
A \define{representative cycle} is a set of simplices that loop around a hole in the complex, and its computation is often of interest, in addition to the one of the associated persistence pair~\cite{bauer2021ripser,Dey_Wang_2022,hang2021umatch,minimal_userguide,obayashi2018volume}.
In the standard algorithm, at the end of the reduction, the nonzero column providing the persistence pair
$(i,j)$ encodes such a representative cycle directly; this is not true in the swap and in the retrospective reductions.
However, it holds that during the execution of either algorithm, once
the persistence pair $(i,j)$ is identified (i.e., before any swapping of column $j$ or any right-to-left column additions on column $j$), the column represents a valid representative for the homology class.
So, while the representatives are not encoded in the final matrix, they can be stored with little extra effort.

\subsection{Further variants}
There are numerous alternatives to obtain reduced columns of (potentially) smaller size. We mention two more variants:
recall that the exhaustive algorithm performs a sequence of further column additions after the pivot has been determined.
In this process, it computes a sequence of columns $c_1,\ldots,c_r$, all with the same pivot and therefore being valid choices for the reduced matrix. 
In the \define{mixed strategy}, we simply remember which column has the smallest size and use its reduced column. 
Since this variant ``locally'' improves the size of a reduced column compared to both the standard and exhaustive variant, one could hope that the mixed strategy improves on both of them.

A (perhaps obvious) further variant is to compute the column with the smallest size among all possible alternatives. 
This problem can be re-phrased as follows.
Given a vector $W$ and $n$ vectors $U_1,\dots,U_n$ in $\mathbb{Z}_2^m$, find $a_1,\dots, a_n$ in $\mathbb{Z}_2$ such that $W + a_1 U_1 + \dots + a_n U_n$ has the minimum number of nonzero  coefficients. 
This problem is called the Nearest Codeword Problem (\textsc{NCW}), which is known to be NP-hard and NP-hard to approximate within any constant factor~\cite{arora}. 
For completeness, we show a simple reduction from \textsc{MaxCut}. 

\begin{proposition}\label{Prop_np_hard}
\textsc{NCW} is NP-hard.
\end{proposition}

\begin{proof}
Let $G = (V,E)$ be a graph, with set of vertices $V=(v_1,\dots,v_n)$ and edges $E=(e_1,\dots,e_m)$, and $M(G) \in \mathbb{Z}_2^{m \times n}$ its vertex-edge incidence matrix. 
Set $W=[1,\dots,1]$, and $U_1,\dots, U_n$ as the columns of $M(G)$.
Given $a_1,\dots,a_n\in \mathbb{Z}_2$, let $W'=W + a_1 U_1 + \dots + a_n U_n$. 
For $e_j=(v_i,v_k)$, we have that $W'[j]=0$ if and only if $a_i=0$ and $a_k=1$, or $a_i=1$ and $a_k=0$. 
Therefore, setting $A=\{i \colon a_i = 0 \}$ and $B=\{i \colon a_i = 1 \}$, we have
\[
\vert \{j \colon W'[j]=0\} \vert = \vert \{e=(v_i,v_k) \in E \colon (i\in A \wedge k \in B) \vee (k\in A \wedge i \in B) \} \vert 
\]
Thus, maximizing the number of zeros of $W'$ is equivalent to finding a maximum cut of $G$, so we can reduce \textsc{MaxCut} to \textsc{NCW}, establishing the NP-hardness of \textsc{NCW}~\cite{karp}.
\end{proof}

\section{Experiments}\label{experiments}
We have implemented our new algorithmic variants (\define{swap}, \define{retrospective}, and \define{mix}) as an extension of the publicly available \phat library~\cite{phat_paper}. 
We also implemented the exhaustive reduction for comparison.
All our algorithms are implemented such that they can be combined with any of the data structures provided by \phat (which required minor extensions of the interface), and 
are included in version 1.7 of the library.
Moreover, we included a branch called \texttt{keeping\_it\_sparse\_experiments} that contains
further modifications and test scripts required to reproduce
the experiments of this section. 
Finally, the datasets
that were used in our experiments together with the required
scripts to produce them are available in a public repository~\cite{repo}.

We address three questions in our experimental evaluation:
\begin{itemize}
\itemsep0em
\item To what extent do our novel approaches really sparsify the reduced boundary matrix, and does this sparsification lead to a reduction in the number of matrix operations performed?
\item What are the most appropriate data structures to represent columns for our novel approaches?
\item How do the best combinations perform in comparison with the default options of \phat?
\end{itemize}

For our tests, we ran our experiments on a workstation with an Intel Xeon E5-1650v3 CPU and 64 GB of RAM, running Ubuntu 18.04.6 LTS, with gcc version 9.4.0 and optimization flags -{}-O3 -{}-DNDEBUG. 
The implementation is not parallelized.

\subsection{Datasets}
We cover different types of filtered simplicial complexes
to investigate the performance in a broader context. 
In particular, we used \define{Vietoris--Rips filtrations} of high-dimensional point clouds,
taken from the benchmark set in~\cite{roadmap} (in all cases, we restricted to the $2$-skeleton without imposing a limit on the edge length), we generated \define{alpha shape filtrations} of random points clouds on a cube, a swissroll and a torus (generated with \cite{tadaset}) and \define{lower star filtrations} generated from publicly available three-dimensional scalar fields~\cite{images_database}. 
The latter are not simplicial but cubical complexes~-- all concepts in this paper carry over to this case without difficulty.

We also include the \define{shuffled filtration}: it is obtained by adding $n$ vertices, then all $\binom{n}{2}$ edges in random order, and finally all $\binom{n}{3}$ triangles in random order.
Such filtrations tend to perform significantly more
column operations than the standard examples.
Therefore, we consider them as a ``stress-test'' for 
challenging reduction tasks that have recently shown up
in the context of image persistence \cite{bauer-schmahl-image} and two-parameter persistence computation using cohomology \cite{bauer-lenzen-lesnick-cohomology}.

\begin{table}[h!]
\begin{adjustbox}{width=\columnwidth,center}
\begin{tabular}[t]{lrrrrrrrrr}
\toprule
\multicolumn{1}{c}{ } & \multicolumn{3}{c}{Cube} & \multicolumn{3}{c}{Swissroll} & \multicolumn{3}{c}{Torus} \\
\cmidrule(l{3pt}r{3pt}){2-4} \cmidrule(l{3pt}r{3pt}){5-7} \cmidrule(l{3pt}r{3pt}){8-10}
Algorithm & fill-in & Col.ops & Bitflips& fill-in & Col.ops & Bitflips& fill-in & Col.ops & Bitflips\\
\midrule 
twist  & 0.56M & 55,613 & 0.23M & 0.62M & 40,257 & 0.15M & 1.10M & 0.10M & 0.50M
\\
twist$^*$  & 0.95M & 45,019 & 0.66M & 1.16M & 41,285 & 0.74M & 1.95M & 39,186 & 1.05M
\\ 
swap  & 0.56M & 55,560 & \bf{0.22M} & 0.62M & \bf{40,243} & \bf{0.14M} & 1.08M & 0.10M & \bf{0.47M}
\\ 
swap$^*$ & 0.85M & \bf{44,547} & 0.54M & 1.02M & 41,293 & 0.61M & 1.74M & \bf{39,184} & 0.84M
\\ 
retro  & 0.17M & 0.52M & 0.76M & 0.20M & 0.59M & 0.87M & \bf{0.31M} & 1.03M & 1.39M
\\ 
retro$^*$  & \bf{0.16M} & 0.38M & 0.45M & \bf{0.18}M & 0.42M & 0.47M & 0.36M & 0.72M & 0.98M
\\ 
exhaust  & 0.19M & 0.81M & 1.42M & 0.21M & 0.99M & 1.69M & 0.32M & 1.67M & 2.68M
\\ 
exhaust$^*$  & 0.20M & 0.42M & 0.60M & 0.21M & 0.45M & 0.56M & 0.38M & 0.82M & 1.22M
\\ 
mix  & 0.18M & 1.08M & 1.97M & 0.21M & 1.35M & 2.41M & 0.32M & 2.27M & 3.89M
\\
mix$^*$  & 0.20M & 0.56M & 0.88M & 0.21M & 0.57M & 0.79M & 0.38M & 3.07M & 5.72M
\\
\bottomrule
\end{tabular}
\end{adjustbox}
\begin{adjustbox}{width=\columnwidth,center}
\begin{tabular}[t]{lrrrrrrrrr}
\toprule
\multicolumn{1}{c}{ } & \multicolumn{3}{c}{Random 50} & \multicolumn{3}{c}{Random 100} & \multicolumn{3}{c}{Senate} \\
\cmidrule(l{3pt}r{3pt}){2-4} \cmidrule(l{3pt}r{3pt}){5-7} \cmidrule(l{3pt}r{3pt}){8-10}
Algorithm & fill-in & Col.ops & Bitflips& fill-in & Col.ops & Bitflips& fill-in & Col.ops & Bitflips\\
\midrule 
twist  & 3,728 & 0.28M & 0.91M & 14,975 & 5.35M & 19.02M & 15,683 & 1.66M & 5.05M
\\
twist$^*$  & 62,780 & \bf{101} & \bf{6,514} & 0.51M & \bf{222} & 38,342 & 0.54M & \bf{124} & 24,729
\\ 
swap   & 3,691 & 0.27M & 0.85M & 14,887 & 5.19M & 17.29M & 15,673 & 1.66M & 5.00M
\\ 
swap$^*$   & 62,420 & \bf{101} & \bf{6,154} & 0.51M & 224 & \bf{33,830} & 0.54M & \bf{124} & \bf{16,655}
\\ 
retro   & \bf{1,274} & 57,677 & 60,124 & \bf{5,049} & 0.48M & 0.49M & \bf{5,355} & 0.53M & 0.54M
\\ 
retro$^*$  & 0.29M & 2,466 & 0.62M & 6.41M & 9,944 & 14.32M & 1.66M & 10,421 & 3.38M 
\\ 
exhaust  & 1,326 & 61,555 & 67,968 & 5,172 & 0.51M & 0.55M & 5,377 & 0.54M & 0.56M
\\ 
exhaust$^*$  & 0.28M & 2,503 & 0.63M & 5.36M & 10,025 & 13.67M & 1.66M & 10,430 & 3.40M
\\ 
mix   & 1,326 & 64,008 & 72,874 & 5,172 & 0.52M & 0.57M & 5,377 & 0.55M & 0.59M
\\
mix$^*$   & 58,693 & 22,753 & 1.12M & 0.49M & 0.22M & 21.84M & 0.52M & 0.19M & 18.83M
\\
\bottomrule
\end{tabular}
\end{adjustbox}
\begin{adjustbox}{width=\columnwidth,center}
\begin{tabular}[t]{lrrrrrrrrr}
\toprule
\multicolumn{1}{c}{ } & \multicolumn{3}{c}{Nucleon} & \multicolumn{3}{c}{Fuel} & \multicolumn{3}{c}{Tooth} \\
\cmidrule(l{3pt}r{3pt}){2-4} \cmidrule(l{3pt}r{3pt}){5-7} \cmidrule(l{3pt}r{3pt}){8-10}
Algorithm & fill-in & Col.ops & Bitflips& fill-in & Col.ops & Bitflips& fill-in & Col.ops & Bitflips\\
\midrule 
twist   & 1.52M & \bf{0.26M} & \bf{1.22M} & 30.15M & 13.29M & 53.28M & 34.70M & 4.90M & 29.43M
\\
twist$^*$   & 1.55M & 0.33M & 2.19M & 31.56M & 13.86M & 89.69M & 33.18M & 4.92M & 30.07M
\\ 
swap   & 1.51M & 0.27M & 1.23M & 30.14M & 13.29M & 53.28M & 33.61M & \bf{4.83M} & \bf{26.79M}
\\ 
swap$^*$   & 1.45M & 0.31M & 1.71M & 24.94M & \bf{13.21M} & 78.44M & 31.59M & 4.86M & 27.15M
\\ 
retro    & 0.41M & 1.04M & 2.16M & 1.51M & 16.73M & 49.55M &  8.13M & 21.51M & 34.70M
\\ 
retro$^*$   & \bf{0.30M} & 1.00M & 1.41M & \bf{1.03M} & 16.09M & \bf{29.64M} & \bf{7.35M} & 21.94M & 31.73M
\\ 
exhaust   & 0.59M & 1.28M & 3.15M & 14.89M & 30.53M & 106.36M & 11.07M & 27.03M & 52.18M
\\ 
exhaust$^*$   & 0.53M & 1.25M & 2.21M & 14.31M & 29.13M & 69.12M & 11.05M & 28.55M & 53.08M
\\ 
mix   & 0.46M & 10.89M & 22.25M & 2.25M & 111.02M & 260.13M & 10.85M & 120.02M & 239.17M
\\
mix$^*$   & 0.43M &  3.79M & 7.40M & 2.02M & 48.46M & 108.28M & 10.79M & 76.24M & 150.20M
\\
\bottomrule
\end{tabular}
\end{adjustbox}
\caption{Fill-in, number of column operations and number of bitflips for (top) alpha shapes filtrations of $10000$ points, sampled respectively from a cube, a swissroll, and a torus in $\mathbb{R}^3$, generated using \cite{tadaset} (each value is the average of 5 random samplings from each shape); (middle) a Vietoris--Rips filtration up to degree $2$ of 50 random points in $\mathbb{R}^{16}$, 100 random points in $\mathbb{R}^4$, and 102 points in $\mathbb{R}^{60}$ of the senate dataset (taken from~\cite{roadmap}); (bottom) the lower star filtrations of the nucleon ($41\times41\times41$ voxels, 68 KB), fuel ($64\times64\times64$ voxels, 256 KB), and tooth ($103\times94\times161$ voxels, 1.5 MB) image from~\cite{images_database}. ``M'' stands for millions, $\ast$ for the dualized matrix. The best performance in each column is in bold.}
\label{tbl:statistics_table3}
\end{table}

\begin{table}[h!]
\begin{adjustbox}{width=\columnwidth,center}
\begin{tabular}[t]{lrrrrrrrrr}
\toprule
\multicolumn{1}{c}{ } & \multicolumn{3}{c}{50 points} & \multicolumn{3}{c}{75 points} & \multicolumn{3}{c}{100 points} \\
\cmidrule(l{3pt}r{3pt}){2-4} \cmidrule(l{3pt}r{3pt}){5-7} \cmidrule(l{3pt}r{3pt}){8-10}
Algorithm & fill-in & Col.ops & Bitflips& fill-in & Col.ops & Bitflips& fill-in & Col.ops & Bitflips\\
\midrule 
twist   & 13,933 & 1.48M & 40.61M & 56,424 & 12.13M & 794.68M & 0.16M & 53.91M & 6,459.57M
\\
twist$^*$  & 1.10M & 11,510 & 21.07M & 9.58M & 58,050 & 497.59M & 44.11M & 184,441 & 4,532.40M 
\\ 
swap & 3,799 & 0.74M & 3.41M & 8,610 & 4.36M & 28.19M & 15,401 & 15.81M & 154.84M
\\ 
swap$^*$  & 0.57M & \bf{8,283} & 9.06M & 5.05M & \bf{39,071} & 249.32M & 22.49M & \bf{0.12M} & 2,334.65M
\\ 
retro  & \bf{1,274} & 81,191 & \bf{0.28M} & \bf{2,849} & 0.34M & \bf{3.28M} & \bf{5,049} & 0.97M & \bf{20.89M}
\\ 
retro$^*$  & 3.81M & 15,563 & 70.71M & 32.70M & 76,790 & 1,582.02M & 147.23M & 0.25M & 13,599.27M
\\ 
exhaust  & 12,735 & 1.10M & 22.16M & 60,825 & 9.58M & 507.11M & 0.19M & 44.10M & 4,576.32M
\\ 
exhaust$^*$ & 1.48M & 12,708 & 39.14M & 12.13M & 53,649 & 782.60M & 53.92M & 0.16M & 6,405.81M
\\ 
mix  & 8,278 & 1.38M & 23.00M & 35,264 & 11.48M & 520.24M & 0.11M & 51.59M & 4,681.67M
\\
mix$^*$  & 1.07M & 30,183 & 36.07M & 9.43M & 0.15M & 741.30M & 43.60M & 0.45M & 6,151.18M
\\
\bottomrule
\end{tabular}
\end{adjustbox}
\caption{Fill-in, number of column operations, and number of bitflips for the shuffled filtrations on 50, 75, and 100 points. Each value is the average of 5 iterations. ``M'' stands for millions, $\ast$ for the dualized matrix. The best performance in each column is in bold.}
\label{tbl:statistics_table_shuffled}
\end{table}

\subsection{Sparsity and bitflips}
We examine the number of nonzero entries of each variant's final reduced boundary matrices, called the \define{fill-in}.
Moreover, we count the number of column additions for each variant to see the effect on efficiency. 
For a more detailed picture, we also count the number of \define{bitflips}
of the algorithm: when adding a column $R_k$ to $R_\ell$, the size of $R_k$ equals the number
of entries in $R_\ell$ that needs to be flipped, and the number of bitflips is the accumulated
number of such flips over all column additions. 
We remark that the name ``bitflip'' is a slight abuse, as it refers to a model where each boundary matrix entry is stored as $1$ bit, which is not necessarily what happens in practice, but it is nevertheless a descriptive name, which is why we chose it.
We expect the number of bitflips to be a good indicator of the practical performance of an algorithm, as the bulk of the running time of matrix reduction is usually spent on column additions. 

The choice of the column representation has no influence on this experiment. 
On the contrary, the number changes dualizing the input matrix. 
Therefore, we tested each algorithm both on the primal and on the dual matrix.

\cref{tbl:statistics_table3,tbl:statistics_table_shuffled} show the outcome for one instance per filtration type. 
We can see that swap reduction consistently leads to smaller reduced matrices and also reduces 
the number of column operations and bitflips, compared to twist reduction (with the exception
of dualized Vietoris--Rips filtrations, where the number of column operations is very small).
The difference is sometimes marginal, though. 
Consistently across the datasets, we can see that the retrospective, followed closely by the mix, outputs much sparser matrices, with the notable exception of dualized Vietoris--Rips.
However, the number of column operations and bitflips is generally not decreasing. 
The exhaustive reduction is also producing similarly sparse matrices, but for shuffled filtrations. 
However, the number of bitflips seems to be generally higher than for the retrospective.
Note that in some examples, the ratio of column operations and bitflips for exhaustive and retrospective is close to one (e.g., in the middle part of Table~1), meaning that the majority of column operations add columns with a single non-zero entry. This shows that the compression strategy is particularly beneficial for these instances.
Finally, the numbers indicate that mix reduction is not a successful strategy: even if it manages to obtain sparse reduced matrices, often sparsifying comparably to the retrospective, it requires many more column operations and many more bitflips (to an extent that surprised the authors).

Hence, our novel variants do improve sparsity quite consistently, but this does not automatically lead to improved performances. 
According to these experiments, there is no direct correlation between sparser matrices and fewer bitflips.

\subsection{Data structures}
We consider the runtime next. 
In particular, we look at the influence of the column representation on the performance
of the algorithm. 
For that, we run each of our algorithms with each of the $8$ available representations in \phat (we refer to \cite{phat_paper} for an extensive description of the data structures).
We show the running times for an alpha shape filtration and a Vietoris--Rips filtration in \cref{celegans_datatype}, for a lower star filtration and shuffled filtration in \cref{shuffled_datatype}. 
Note that these tables show only the running time for the matrix reduction; the time to read the input file into memory and to (potentially) dualize the matrix (both of which usually take more time than the reduction itself) is not shown. 

\begin{table}[h!]
\begin{adjustbox}{width=\columnwidth,center}
\begin{tabular}[t]{lrrrrrrrr}
& List & Vector & Set & Heap & A-Heap & A-Set & A-Full & A-Bit-Tree 
\\
\hline
twist  & 0.3 & 0.1 & 0.2 & 0.2 & 0.2 & 0.2 & 0.2 & \bf{0.1} 
\\
twist$^*$ & 7.2 & 0.3 & 0.6 & 0.4 & 0.4 & 0.5 & 0.3 & 0.2
\\
swap  & 0.3 & \bf{0.1} & 0.2 & 0.4 & 0.4 & 0.2 & 0.2 & 0.4
\\
swap$^*$ & 9.4 & 0.3 & 0.6 & 21.4 & 10.0 & 0.5 & 0.3 & 5.5 
\\
retro  & 0.7 & 0.6 & 0.8 & 1.0 & 0.7 & 0.7 & 0.7 & 0.7 
\\
retro$^*$ & 0.7 & \bf{0.5} & 0.8 & 0.8 & 0.6 & 0.6 & 0.6 & 0.6  
\\
exhaust  & 0.7 & 0.5 & 0.9 & 0.8 & 0.8 & 0.7 & 0.7 & 0.7 
\\
exhaust$^*$ & 0.7 & \bf{0.4} & 0.7 & 0.6 & 0.6 & 0.6 & 0.5 & 0.5  
\\
mix  & 0.9 & 0.6 & 1.1 & 1.0 & 1.0 & 1.0 & 1.0 & 1.0  
\\
mix$^*$ & 0.8 & \bf{0.5} & 1.0 & 0.8 & 0.8 & 0.7 & 0.8 & 0.8 
\\
\bottomrule
\end{tabular}
\end{adjustbox}
\begin{adjustbox}{width=\columnwidth,center}
\begin{tabular}[t]{lrrrrrrrr}
& List & Vector & Set & Heap & A-Heap & A-Set & A-Full & A-Bit-Tree 
\\
\hline
twist & 3.9 & 1.5 & 3.4 & 3.8 & 3.5 & 3.5 & 2.2 & 2.4 
\\
twist$^*$ & 0.7 & \bf{0.1} & \bf{0.1} & 0.3 & \bf{0.1} & 0.2 & \bf{0.1} & \bf{0.1} 
\\
swap & 4.0 & 1.6 & 3.4 & 5.1 & 4.1 & 3.6 & 2.5 & 9.0 
\\
swap$^*$ & 0.8 & \bf{0.1} & \bf{0.1} & 1.2 & 0.8 & 0.2 & \bf{0.1} & 0.5 
\\
retro  & 1.7 & \bf{1.4} & 1.9 & 3.3 & 2.0 & 2.5 & 1.8 & 2.0 
\\
retro$^*$ & 73.9 & 6.1 & 42.8 & 155.8 & 191.7 & 237.3 & 120.6 & 38.7 
\\
exhaust & 1.3 & \bf{0.9} & 1.5 & 1.8 & 1.7 & 1.4 & 1.3 & 1.5 
\\
exhaust$^*$ & 46.0 & 1.7 & 12.1 & 13.5 & 12.6 & 11.7 & 5.7 & 2.3 
\\
mix & 1.3 & \bf{0.9} & 1.4 & 1.8 & 1.4 & 1.4 & 1.4 & 1.4 
\\
mix$^\ast$ & +5m & 186.1 & 81.1 & +5m & +5m & 65.0 & 20.5 & +5m
\\
\bottomrule
\end{tabular}
\end{adjustbox}
\caption{(Top) Alpha filtrations on 40,000 points on a cube, average reduction time over 5 random samplings. (Bottom) Vietoris--Rips filtration, 297 points in ambient dim 202 up to degree $2$ (celegans dataset from~\cite{roadmap}). All timings are in seconds but for the timeout (minutes), $\ast$ stands for the dualized matrix, and the best runtime(s) per algorithm (both over primal and dual input) is in bold.}
\label{celegans_datatype}
\end{table}

\begin{table}[h!]
\begin{adjustbox}{width=\columnwidth,center}
\begin{tabular}[t]{lrrrrrrrr}
& List & Vector & Set & Heap & A-Heap & A-Set & A-Full & A-Bit-Tree 
\\
\hline
twist  & +5m & 12.9 & 3.0 & 3.0 & 3.0 & 3.4 & 2.3 & 2.0
\\
twist$^*$ & +5m & +5m & 4.8 & 4.8 & 4.0 & 4.2 & 2.5 & \bf{1.9}
\\
swap  & +5m & 12.9 & 2.8 & +5m & +5m & 3.4 & \bf{2.4} & +5m 
\\
swap$^*$ & +5m & +5m & 4.8 & +5m & +5m & 4.2 & 2.6 & +5m
\\
retro  & 5.6 & \bf{4.7} & 6.2 & 9.1 & 6.9 & 7.2 & 6.6 & 7.1 
\\
retro$^*$ & 6.8 & 4.9 & 7.1 & 9.4 & 7.2 & 7.7 & 7.0 & 7.5  
\\
exhaust  & 5.8 & \bf{4.1} & 6.7 & 6.6 & 6.9 & 6.6 & 5.9 & 6.3 
\\
exhaust$^*$ & 7.8 & 4.7 & 8.4 & 7.1 & 7.3 & 7.2 & 6.3 & 6.7  
\\
mix  & 14.0 & 9.0 & 15.8 & 15.6 & 16.2 & 15.6 & 13.9 & 18.6  
\\
mix$^*$ & 14.0 & \bf{8.1} & 15.2 & 13.8 & 13.4 & 13.0 & 11.7 & 14.5
\\
\bottomrule
\end{tabular}
\end{adjustbox}
\begin{adjustbox}{width=\columnwidth,center}
\begin{tabular}[t]{lrrrrrrrr}
& List & Vector & Set & Heap & A-Heap & A-Set & A-Full & A-Bit-Tree 
\\
\hline
twist  & 28.6 & 6.8 & 58.2 & 59.4 & 54.8 & 46.8 & 6.8 & 2.1
\\
twist$^*$ & 47.4 & 4.4 & 69.2 & 55.9 & 52.6 & 41.1 & 4.7 & \bf{1.3} 
\\
swap  & 2.8 & \bf{0.5} & 2.4 & 6.9 & 4.9 & 1.7 & 0.6 & 3.4 
\\
swap$^*$ & 28.3 & 1.9 & 38.0 & 86.5 & 56.0 & 22.9 & 6.7 & 14.2
\\
retro  & 0.2 & \bf{0.1} & 0.3 & 0.7 & 0.6 & 0.7 & 0.4 & 0.2
\\
retro$^*$ & 100.1 & 17.2 & 257.9 & 432.7 & 452.0 & 420.5 & 271.6 & 55.7 
\\
exhaust  & 17.8 & 4.4 & 36.6 & 36.4 & 33.8 & 29.3 & 4.4 & \bf{1.4}
\\
exhaust$^*$ & 51.9 & 6.5 & 85.6 & 87.5 & 82.7 & 64.5 & 8.0 & 2.2 
\\
mix  & 20.5 & 4.8 & 37.8 & 37.8 & 35.0 & 30.3 & 4.8 & \bf{1.7}   
\\
mix$^*$ & 143.4 & 7.9 & 117.8 & 211.9 & 158.0 & 65.4 & 89.2 & 29.7
\\
\bottomrule
\end{tabular}
\end{adjustbox}
\caption{(Top) Lower star filtration for the tooth image from~\cite{images_database} ($104\times 91\times 161$ voxels, 1.5 MB). (Bottom) Average reduction times for 5 iteration of a shuffled filtration on 75 points. All timings are in seconds but for the timeout (minutes), $\ast$ stands for the dualized matrix, and the best runtime per algorithm (both over primal and dual input) is in bold.}
\label{shuffled_datatype}
\end{table}

First of all, the tables confirm the earlier findings in~\cite{phat_paper}: the performance of the twist algorithm
highly varies depending on the chosen data structure, and the best results are achieved using the A-Bit-Tree representation, which is also the default in \phat.
Moreover, for Vietoris--Rips filtrations, it is highly beneficial to dualize the matrix.

The swap reduction generally performs similarly to twist, working fast with dualization on Vietoris--Rips complexes. 
Remarkably, it performs much better if run with A-Full on the lower star filtration.
Swap is generally slower on the data structures Heap, A-Heap, and A-Bit-Tree. 
The explanation lies in the way how we use these data structures
to represent columns, which does not permit a constant-time access to the size of a column. 
As the size is queried before every column addition, this results in a considerable slow-down for the swap algorithm.

We observe that the performance for the retrospective algorithm seems more stable across different data structures than for other algorithms on alpha and lower star filtrations. 
We speculate that the reason for this general stability is the general sparsity of the columns, which reduces the importance of how the entries are stored in memory.
We also observe that it sometimes, but not always, improves on the exhaustive algorithm.
Remarkably, the retrospective is competitive in practice (for example, on the alpha filtration) even if it is not implemented with the clearing optimization. 
We are not aware of other variants with this characteristic.

As expected from \cref{tbl:statistics_table3,tbl:statistics_table_shuffled}, the mix algorithm generally has the poorest practical performance among all tested algorithms. 
We therefore leave it out in further comparisons.

Based on our experiments, we identified that twist works most efficiently in combination with A-Bit-Tree (as previously known), and the retrospective and the exhaustive algorithms work best with the Vector representation but for the shuffled filtrations, where exhaustive should be paired with A-Bit-Tree. 
The swap reduction sometimes works best with Vector and sometimes with A-Full. 
Since the advantage of A-Full was generally more significant, we chose A-Full for subsequent experiments. 

\subsection{Performance on large datasets}
We now compare the performance on larger instances.
We focus on the combinations that were identified to be most efficient in the previous experiments. 
We run each algorithm on the original matrix and its dual and pick the better of the two runtimes.

The results are displayed in \cref{tbl:runtime_table}. 
We observe that retrospective is systematically by 1 or 2 orders of magnitude better than twist for shuffled, and worse by one or two orders of magnitude for Vietoris--Rips filtrations.
Also, it is interesting that whenever twist is faster after dualizing, retrospective is faster without dualization and vice versa. 
We observe that swap is always as fast or slightly slower than the twist on the first three filtrations. 
It is decidedly faster only on shuffled filtrations; however, it is not as fast as retrospective. 
Moreover, swap and twist mostly perform better on the same type of matrices: original for alpha shape and lower star filtrations, and dual on Vietoris--Rips. 
Remarkably, swap performs better on shuffled filtrations without dualization, unlike the twist reduction. 
Interesting is the case of lower star filtrations: while here it is clear that swap and exhaustive are not the best performing, it is more difficult to choose between twist and retrospective. 
Indeed, while retrospective is generally better than twist, sometimes the latter takes only half the time. 
This behavior does not correlate with the input sizes. 
We theorize that it depends on the number of short-lived bars: when there are more, the sparsification of the retrospective pays off (as it happens in the \textit{bonsai} image), while when there are fewer, like in the \textit{skull} image, it does not.
Finally, the exhaustive scales as retrospective on alpha shape filtrations but is slower by two orders of magnitude on the shuffled filtration. 
Moreover, exhaustive performs slightly better than retrospective on Vietoris--Rips filtrations, even if not as good as swap and twist. 

\begin{table}[h!]
\begin{adjustbox}{width=\columnwidth,center}
\begin{tabular}[t]{lrrrrrrrrr}
\toprule
\multicolumn{1}{c}{ } & \multicolumn{3}{c}{Cube} & \multicolumn{3}{c}{Swissroll}& \multicolumn{3}{c}{Torus} \\
\cmidrule(l{3pt}r{3pt}){2-4} \cmidrule(l{3pt}r{3pt}){5-7} \cmidrule(l{3pt}r{3pt}){8-10}
Algorithm & 40K & 80K & 160K & 40K & 80K & 160K & 40K & 80K & 160K \\
\midrule
twist+A-Bit-Tree  & \textbf{0.1} & \textbf{0.3} & \textbf{0.8} & \textbf{0.1} & \textbf{0.2} & \textbf{0.5} & \textbf{0.2} & \textbf{0.4} & \textbf{0.8}
\\
swap+A-Full  & 0.2 & 0.4 & 0.9 & \textbf{0.1} & \textbf{0.2} & \textbf{0.5} & 0.3 & 0.6 & 1.5
\\
retro+Vector  & $^\ast$0.5 & $^\ast$1.2 & $^\ast$2.7 & $^\ast$0.5 & $^\ast$1.2 & $^\ast$2.7 & $^\ast$1.0 & $^\ast$2.3 & $^\ast$5.4
\\
ex+Vector & $^\ast$0.4 & $^\ast$1.0 & $^\ast$2.3 & $^\ast$0.4 & $^\ast$1.0 & $^\ast$2.3 & $^\ast$0.8 & $^\ast$1.8 & $^\ast$4.2
\\
\bottomrule
\end{tabular}
\end{adjustbox}
\begin{adjustbox}{width=\columnwidth,center}
\begin{tabular}[t]{lrrrrrrrrr}
\toprule
Algorithm  & Hydrogen & Shockwave & Lobster & MRI Head & Engine & Statue Leg & Bonsai & Skull \\
\midrule 
twist+A-Bit-Tree  & 12.3 & 12.2 & \textbf{25.8} & \textbf{13.8} & 54.5 & 53.9 & 177.9 & \textbf{24.2}
\\
swap+A-Full & 17.3 & 15.3 & 37.8 & 14.9 & 74.3 & 71.0 & 264.4 & 25.0 
\\
retro+Vector  & $^\ast$\textbf{10.9} & $^\ast$\textbf{10.4} & $^\ast$30.3 & $^\ast$30.3 & $^\ast$\textbf{54.4} & $^\ast$\textbf{50.8} & $^\ast$\textbf{150.6} & $^\ast$62.4 
\\
ex+Vector  & $^\ast$28.6 & $^\ast$21.9 & $^\ast$62.2 & $^\ast$29.3 & $^\ast$239.2 & $^\ast$102.5 & $^\ast$+5m & $^\ast$58.3 
\\
\bottomrule
\end{tabular}
\end{adjustbox}
\begin{adjustbox}{width=\columnwidth,center}
\begin{tabular}[t]{lrrrrrrrrrrrr}
\toprule
\multicolumn{1}{c}{ } & \multicolumn{6}{c}{Vietoris--Rips}& \multicolumn{3}{c}{Shuffled} \\
\cmidrule(l{3pt}r{3pt}){2-7} \cmidrule(l{3pt}r{3pt}){8-10} 
Algorithm & 297 & 300 & 445 & 512 & 1000 & 1000 & 100 & 125 & 150
\\
\midrule
twist+A-Bit-Tree & \textbf{$^\ast$0.1} & \textbf{$^\ast$0.1} & \textbf{$^\ast$0.3} & \textbf{$^\ast$0.5} & \textbf{$^\ast$3.2} & \textbf{$^\ast$3.0} & $^\ast$8.4 & $^\ast$39.2.5 & $^\ast$154.4
\\
swap+A-Full & \textbf{$^\ast$0.1} & \textbf{$^\ast$0.1} & \textbf{$^\ast$0.3} & $^\ast$0.8 & $^\ast$4.4 &$^\ast$3.6 & 2.7 & 8.1 & 23.4
\\
retro+Vector & 1.6 & 2.1 & 8.1 & 15.4 & 170.7 & 171.0 & \textbf{0.3} & \textbf{0.9} & \textbf{2.9}
\\
ex+Vector/A-Bit-Tree & 0.9 & 1.2 & 4.8 & 9.3 & 121.2 & 119.5 & 10.8 & 50.0 & 192.8
\\
\bottomrule
\end{tabular}
\end{adjustbox}
\caption{Running times (in seconds) on alpha shape filtrations on 40K, 80K, and 160K points sampled from a cube, a swissroll, and a torus data sets, (top) and various data sets for the lower star, Vietoris--Rips, and shuffled filtrations. The images are: hydrogen ($128\times 128\times 128$ voxels, 2 MB),
shockwave ($64\times 64\times 512$, 2.0 MB), lobster ($301\times324\times56$ voxels, 5.2 MB), MRI head ($256\times 256\times 124$ voxels, 7.8 MB), engine ($256\times 256\times 128$ voxels, 8 MB), statue leg ($341\times 341 \times 93$ voxels, 10.3 MB), bonsai ($256\times 256\times 256$ voxels, 16 MB), and skull ($256\times256\times256$ voxels, 16.0 MB) from \cite{images_database}. The Vietoris--Rips datasets are the celegans, vicsek 1, house, fractal linear edge, 1000 random points in $R^8$, and dragon from~\cite{roadmap}, ordered by their numbers of points displayed at the top. Each value for the shuffled filtration is the average over 5 samplings. The $^\ast$ signals that the running time was achieved using the dual matrix. The best performance per input is highlighted in bold.}
\label{tbl:runtime_table}
\end{table}
 
In general, we see that our variants keep the matrix sparse for the shuffled filtration and this leads to a significant improvement in efficiency. 
Likely, this happens because, in the reduction of random filtrations, each column is added several times. 
Thus, keeping it sparse pays off.

\subsection{Memory consumption}
We also tested the memory peak consumption of the algorithms (\cref{tbl:memory_table}). 
For the alpha shape filtrations, the twist and the swap are better than the retrospective and exhaustive. 
On the other three filtrations, the exhaustive is generally better than any other algorithms. 
Twist and swap are quite similar on alpha and Vietoris--Rips filtrations, but swap uses considerably less memory on shuffled, possibly because it does not need to dualize. 
In these experiments, the memory overhead for dualizing the input matrix is included in these numbers, which might partially explain why the dualized instances usually take more memory. 

\begin{table}[h!]
\begin{adjustbox}{width=\columnwidth,center}
\begin{tabular}[t]{lrrrrrrrrr}
\toprule
\multicolumn{1}{c}{ } & \multicolumn{3}{c}{Cube} & \multicolumn{3}{c}{Swissroll} & \multicolumn{3}{c}{Torus}
\\
\cmidrule(l{3pt}r{3pt}){2-4} \cmidrule(l{3pt}r{3pt}){5-7} \cmidrule(l{3pt}r{3pt}){8-10} 
Algorithm & 40K & 80K & 160K & 40K & 80K & 160K & 40K & 80K & 160K 
\\
\midrule
twist+A-Bit-Tree  & 0.10G & 0.20G & 0.40G & 0.13G & 0.26G & 0.54G & 0.23G & 0.48G & 1.01G
\\
swap+A-Full & 0.13G & 0.25G & 0.50G & 0.16G & 0.32G & 0.67G & 0.28G & 0.60G & 1.27G
\\
retro+Vector  & $^\ast$0.18G & $^\ast$0.35G & $^\ast$0.71G & $^\ast$0.23G & $^\ast$0.47G & $^\ast$0.97G & $^\ast$0.40G & $^\ast$0.86G & $^\ast$1.86G
\\
exhaustive+Vector & $^\ast$0.18G & $^\ast$0.35G & $^\ast$0.71G & $^\ast$0.23G & $^\ast$0.47G & $^\ast$0.97G & $^\ast$0.40G & $^\ast$0.86G & $^\ast$1.86G
\\
\bottomrule
\end{tabular}
\end{adjustbox}
\begin{adjustbox}{width=\columnwidth,center}
\begin{tabular}[t]{lrrrrrrrrr}
\toprule
Algorithm & Hydrogen & Shockwave & Lobster & MRI Head & Engine & Statue Leg & Bonsai & Skull \\
\midrule 
twist+A-Bit-Tree & 3.17G & 3.20G & 6.70G & 5.78G & 11.88G & 13.44G & 34.99G & 12.02G
\\
swap+A-Full & 3.55G & 3.56G & 7.63G & 7.16G & 13.32G & 15.33G & 37.91G & 14.94G
\\
retro+Vector & $^\ast$9.16G & $^\ast$4.66G & $^\ast$9.11G & $^\ast$10.31G & $^\ast$15.65G & $^\ast$18.47G & $^\ast$42.06G & $^\ast$21.55G
\\
ex+Vector & $^\ast$2.78G & $^\ast$2.84G & $^\ast$6.92G & $^\ast$10.31G & $^\ast$10.73G & $^\ast$13.68G & $^\ast$26.48G & $^\ast$21.55G
\\
\bottomrule
\end{tabular}
\end{adjustbox}
\begin{adjustbox}{width=\columnwidth,center}
\begin{tabular}[t]{lrrrrrrrrrrrr}
\toprule
\multicolumn{1}{c}{ } & \multicolumn{6}{c}{Vietoris--Rips}& \multicolumn{3}{c}{Shuffled} \\
\cmidrule(l{3pt}r{3pt}){2-7} \cmidrule(l{3pt}r{3pt}){8-10} 
Algorithm & 297 & 300 & 445 & 512 & 1000 & 1000 & 100 & 125 & 150
\\
\midrule
twist+A-Bit-Tree & $^\ast$0.66G & $^\ast$0.68G & $^\ast$2.22G & $^\ast$3.41G & $^\ast$25.18G & $^\ast$25.18G & $^\ast$0.37G & $^\ast$1.09G & $^\ast$2.85G
\\
swap+A-Full & $^\ast$0.76G & $^\ast$0.78G & $^\ast$2.54G & $^\ast$3.87G & $^\ast$28.82G & $^\ast$28.82G & 22.51M & 40.86M & 72.05M
\\
retro+Vector & 0.45G & 0.46G & 1.50G & 2.28G & 16.93G & 16.93G & 31.47M & 62.66M & 118.13M
\\
ex+Vector/A-Bit-Tree & 0.35G & 0.36G & 1.15G & 1.75G & 13.05G & 13.05G & 19.17M & 33.69M & 56.10 M
\\
\bottomrule
\end{tabular}
\end{adjustbox}
\caption{Memory peak consumption for the best-performing combination of algorithms, data structures, and dualization. The $^*$ signals that the running time was achieved using the dual matrix. ``G'' and ``M'' stand for gigabyte and megabyte, respectively. The inputs are the datasets from \cref{tbl:runtime_table}.}
\label{tbl:memory_table}
\end{table}

\subsection{Summary of the experiments}
The experiments show that, while the sparsity of the matrix is a necessary condition for efficiency, it is, perhaps surprisingly, not always a sufficient one. 
The general behavior appears to be that the more structured the data, the less it pays off to sparsify.

For highly structured data, such as Vietoris--Rips and alpha filtrations, the price of keeping the matrix sparse exceeds its advantages. 
Indeed, the only competitive sparsifying algorithm in this situation is the swap algorithm, which does only a few extra operations and does not sparsify aggressively. 

On the other hand, for random data such as the shuffled filtration, sparsifying is a winning strategy. 
In this situation, the retrospective, which tries most aggressively to sparsify the matrix, is by two orders of magnitude faster than the twist, and the swap outperforms the twist by one order of magnitude. 

Another emerging behavior is that exhaustive performs quite badly across all input data. 
Possibly, this depends on the fact that while it does many operations, it keeps the matrix less sparse than other sparsifying algorithms. 
In other words, it pays a price to sparsify but the payoff is lower.

Last but not least, the experiments on the lower star filtrations indicate that the winning strategy is to use either the twist or the retrospective, but there is no emerging trend indicating when to use which. 
Even more, usually the retrospective is faster, but when it is not it performs much worse than the twist. 
A possible explanation is that this depends on the number of shorter bars: when there are more, and therefore more computation is required, sparsifying pays off.

\section{Output-sensitive bounds}\label{S_bounds}

The idea of the retrospective reduction is to keep reducing the columns-to-be-added using the newly found pivots. 
This, together with the fact that pivots encode information about persistence pairs, allows us to bound the number of bitflips with the persistence Betti numbers. 
To prove the bounds, we group the column additions into complementary classes: forward/backward, depending on if the addition is left-to-right or right-to-left, and non-/ interval, according to whether the column indices are outside or inside some persistence interval. 
The first bound is then obtained by counting the bitflips from forward additions directly and the ones from backward using interval additions. 
The second bound follows by counting the bitflips from non-interval additions directly and the ones from intervals using forward and backward additions.
\medskip

Let us write $P$ for the set of pivot pairs of a filtered simplicial complex and
\[
\overline{P}:= P\cup\{(i,N+1)\mid \text{$\sigma_i$ is essential}\} \, .
\] 
The \define{Betti number} $\betti{k}$ for $1\leq k\leq N$ is defined as
$\betti{k} := \# \{(i,j)\in \overline{P} \mid i\leq k < j\}\}$.
The topological interpretation is that $\betti{\ell}$ is the number of holes in the complex $\altcomplex_{\ell}$. The \define{persistent Betti number} $\betti{k,\ell}$, for $1\leq k \leq \ell \leq N$ is defined as
\[\betti{k,\ell}:= \# \{(i,j)\in\overline{P}\mid i\leq k, \ j > \ell\} \]
and gives the number of holes that are persistently present in all complexes $\altcomplex_{k},\ldots,\altcomplex_{\ell}$.
In \cref{{algorithm_lr_red}}, $R_k$ is \define{pivoted} after the procedure {\tt Reduce} is invoked with $k$ as the argument in the \textbf{for} loop of the procedure \main.

\begin{observation}\label{obs:terminate}
By construction of the {\tt Reduce} procedure, for any step $\ell$, after any invocation of {\tt Reduce}$(j)$, the entries above $\lowp{R_j}$ are unpaired at $\ell$. In particular, after the first invocation of {\tt Reduce}$(j)$, the number of entries above $\lowp{R_j}$ is bounded above by $\beta_j$.
\end{observation}

The addition of $R_k$ to $R_{\ell}$ is called \define{forward} if $k < \ell$ and \define{backward} if $k > \ell$.
A \define{forward} (resp. \define{backward}) \define{bitflip} is a bitflip resulting from a forward (backward) addition.
For integers $k,\ell$, the bitflip of $R_{\ell}^i$ when adding $R_k$ to $R_{\ell}$ is an \define{interval bitflip} for a pair $(i,j)\in\overline{P}$ if $k, \ell \in \{i+1,\dots, j\}$, and \define{non-interval} otherwise. For the retrospective algorithm, the interval bitflips dominate the cost. Moreover, this cost can be bounded in terms of persistence Betti numbers and index persistence.

\begin{observation}\label{obs:onlyback}
For every $k = 1,\dots, N$, the first iteration of {\tt Reduce$(k)$} finds the pivot, and the subsequent iterations perform only backward additions to $R_k$.
\end{observation}

\begin{lemma}\label{lem:backregular}
If $ k < \ell$ and $R_\ell$ is added to $R_k$, all the resulting bitflips are interval.
\end{lemma}

\begin{proof} 
By Observation~\ref{obs:onlyback}, if the addition of $R_\ell$ to $R_k$ flips the $i$-th entry, then $i \leq \lowp{R_k}$. 
Thus, $i<k<\ell$. 
For the other inequality, by Observation~\ref{obs:terminate}, it follows that when $R_\ell$ is added to $R_k$ the entries in $R_\ell$ are not paired with indices less than $\ell$. 
\end{proof}

\begin{lemma}\label{lem:atmostonce} 
Any two columns are added to each other at most once backward and at most once forward. 
\end{lemma}

\begin{proof}  
Fix $1\leq k < \ell \leq N$.
$R_k$ is added to $R_{\ell}$ only once when $R_{\ell}$ is being pivoted. $R_{\ell}$ is added to $R_k$ only if $\ell$ is paired to $j$ and $R_k^{j}\neq 0$.
Once $R_k^{j}$ is eliminated from $R_k$, it is not reintroduced, since $j$ is now paired, and therefore $R_{\ell}$ is added to $R_k$ at most once.
\end{proof}

\begin{lemma} \label{lem:eternalbound}
Let $(i,j)$ be a pivot pair. 
Once $R_{j}$ is pivoted, $\sizep{R_{j}}$ is always $\leq 1 + \persbetti{i,j}$.
\end{lemma}

\begin{proof}
By Observation~\ref{obs:terminate}, immediately after $R_j$ is pivoted,  the entries above $i=\lowp{R_j}$ are unpaired at $j$, and hence  $\sizep{R_{j}} \leq 1 + \persbetti{i,j} $.
Subsequently, $R_\ell$ is added to $R_j$ only if $\ell > j$ and $\lowp{R_{\ell}} <i $. 
Moreover, when $R_\ell$ is added to $R_j$, by Observation~\ref{obs:terminate}, the entries above $R_\ell$ are unpaired at $\ell$. 
Hence, $\sizep{R_{j}} \leq 1 + \persbetti{i,j}$ is maintained in subsequent additions into $R_j$.
\end{proof}

As an immediate corollary, we have: 

\begin{corollary} \label{cor:bothwaybounds}
Let $(i,j)\in P$. The number of bitflips of an addition to or from $R_j$ is~$\leq 1+\beta_{i,j}$.
\end{corollary}

\begin{proposition}  \label{prop:mainbounds}
The total number of bitflips in \cref{algorithm_lr_red} is bounded by
\[
\sum_{(i,j) \in P} (\betti{i,j} + 1) \cdot \min \{\betti{i,j} + 1, j-i+1\} + \sum_{k = 1}^N (d_k+1)\left(\betti{k}+1\right) \, .
\]
\end{proposition}

\begin{proof}
The first term is obtained by bounding the backward bitflips in two fashions: \emph{into} and \emph{from} $R_j$. 
Fix a pivot pair $(i,j)$. 
By \cref{cor:bothwaybounds} any additions involving $R_j$ has at most $1+\beta_{i,j}$ bitflips, providing the first factor. 
We now count how many columns are added to $R_j$, how many columns $R_j$ is added to, and take the minimum. 
Backward additions into $R_j$ are executed to zero out a formerly unpaired entry that is now paired, and at most $\persbetti{i,j}+1$ entries from $R_j$ need to be zeroed out.
By \cref{lem:backregular}, $R_j$ is added only to columns $R_{\ell}$ for $i+1\leq \ell \leq j-1$, and by \cref{lem:atmostonce} the backward additions happen at most once. 

The second term bounds forward bitflips.
Fix $\ell < k$. 
By Observation~\ref{obs:terminate}, before $R_\ell$ is added to $R_k$, $\sizep{R_\ell} \leq 1+\betti{k}$. 
The paired entries in $R_k$ before $R_k$ is pivoted is bounded by $d_k+1$.
Hence, $d_k+1$ is the maximum number of columns that need to be added to $R_k$ to zero out, and each is added only once (\cref{lem:atmostonce}).
So, the total number of forward bitflips in $R_k$ is bounded by $(d_k+1)(\betti{k}+1)$.
\end{proof}

\begin{proposition}  \label{bigtwoth}
The total number of bitflips required to reduce $D$ in \cref{algorithm_lr_red} is bounded by 
\[
 \sum\limits_{  {(i,j) \in \overline{P}}} ( j - i)^2 \,\, + \,\, \sum_{k=1}^{N}(d_k+1) \, .
\]
\end{proposition}

\begin{proof} 
The two addends are bound, respectively, by the interval and non-interval bitflips.

Let $(i,j) \in \overline{P}$. 
For $i < k < \ell \leq j$, by \cref{lem:atmostonce}, the elements $R_{k}^i$ and $R_{\ell}^i$ are added to each other at most once. 
So, the total number of forward and backward interval bitflips in row $i$ is bounded by $(j-i)^2/2$ each, and the first term follows.

By \cref{lem:backregular}, all non-interval bitflips are forward bitflips.
By Observation~\ref{obs:terminate}, if $R_k$ is added to $R_\ell$ for $k<\ell$, the entries above $\lowp{R_k}$ are not paired until $\ell$, and  lead to interval bitflips in $R_\ell$. 
As a result of adding $R_k$, the only non-interval bitflips in $R_\ell$ occur in the row index of $\lowp{R_k}$.
Since $D$ is a boundary matrix, at the beginning every column $k$ has at most $d_k+1$ unpaired entries which need to be zeroed out, proving the claim.
\end{proof}

If we restrict the matrix $D$ to the $p$-simplices of the complex, we obtain the $p$-dimensional boundary matrix $D^{(p)}$.
We then have a finer analysis of \cref{bigtwoth}. 
The bound for the number of bitflips required to reduce $D^{(p)}$ is
\[ 
\sum\limits_{  {(i,j) \in \overline{P}}, \,\, d_j = p} ( j - i)^2 \,\, + \,\, N(p+1) \, .
\]

Using Observation~\ref{obs:terminate} and \cref{lem:backregular}, the maximum number of entries in $R_i$ for $i=1,\dots,N$ during the course of the algorithm  is $d_i + \betti{i} + 1$. 
Let $\beta_{\max} = \max_{i=1,\dots,N} \betti{i}$. 
Then the peak memory consumption for \cref{algorithm_lr_red} is bounded by $O(N (\max_i d_i+\beta_{\max}))$.

\section{Differentiating examples}
\label{sec:differentiating}
Our experiments have shown that, in practice, retrospective and swap reductions have the potential to run faster than twist.
However, there exist constructions for which either of the three mentioned algorithms performs asymptotically better than the other two. 
Specifically:

\begin{proposition}\label{prop:differentiate}
Let $A\in\{\text{twist},\text{swap},\text{retrospective}\}$. Then there exists an infinite family of filtered
simplicial complexes with increasing size $n$, such that the number of bitflips for matrix reduction using
$A$ is bounded by $O(n)$, where the number of bitflips for the other two algorithms is $\Omega(n^2)$.
\end{proposition}
We prove this statement by four constructions of filtered simplicial complexes with the following properties:
\begin{itemize}
\itemsep0em 
\item Complex $K_1$ causes $O(n)$ bitflips for retrospective, and $\Omega(n^2)$ bitflips for twist and swap, 
\item Complex $K_2$ causes $O(n)$ bitflips for twist and swap, and $\Omega(n^2)$ bitflips for retrospective.
\item Complex $K_3$ causes $O(n)$ bitflips for swap, and $\Omega(n^2)$ bitflips for twist.
\item Complex $K_4$ causes $O(n)$ bitflips for twist, and $\Omega(n^2)$ bitflips for swap.
\end{itemize}
The statement follows directly from these $4$ constructions.

\begin{figure}[h!]
\centering
\begin{subfigure}[h]{.5\linewidth}
\begin{tikzpicture}[line cap=round,line join=round,>=triangle 45,x=1cm,y=1cm]
\clip(-3.0106089776135723,-1.6549568379899828) rectangle (4.611997426436136,2.7932220949215076);
\begin{scriptsize}
\begin{pgfonlayer}{foreground}
\node[dot=3pt] (v0) at (0,0) {};
\node[dot=3pt] (v1) at (0.7865072033647845,2.37305845251551) {};
\node[dot=3pt] (v1) at (0.7865072033647845,2.37305845251551) {};
\node[dot=3pt] (v2) at (1.4945672579355689,2.0064266280775427) {};
\node[dot=3pt] (v3) at (2.0485511735311537,1.4325014320868856) {};
\node[dot=3pt] (v4) at (-2.278994628732195,1.0277890074745213) {};
\node[dot=3pt] (v5) at (-1.8399580316361015,1.6924994658249684) {};
\node[dot=3pt] (v6) at (-0.8213486526906792,2.688734670322916) {};
\node[dot=3pt] (v7) at (-1.0054455436501044,2.638978753847396) {};
\node[dot=3pt] (v8) at (-1.7169551492500446,2.2309802387481317) {};
\end{pgfonlayer}
\draw [line width=1pt] (v0) -- (v1);
\draw [line width=1pt] (v0) -- (v2);
\draw [line width=1pt] (v0) -- (v3);
\draw [line width=1pt] (v2) -- (v1);
\draw [line width=1pt] (v2) -- (v3);
\draw[color=black] (1.2,2.3) node {$e_2$};
\draw[color=black] (1.9,1.8) node {$e_3$};
\draw[color=black] (0.3,1.55) node {$e_1$};
\draw[color=black] (1.3,1.2) node {$e_{n+2}$};
\draw[color=black] (1.4,0.6) node {$e_{n+3}$};
\draw[color=black] (0.9,1.7) node {$t_n$};
\draw[color=black] (1.55,1.45) node {$t_{n+1}$};
\draw[color=black] (-2.4,1.4) node {$e_{n+1}$};
\draw[color=black] (-1.4,1) node {$t_1$};
\begin{scope}[on background layer] 
\fill [gray!20] (v0.center) -- (v1.center) -- (v2.center) -- cycle;
\fill [gray!20] (v0.center) -- (v2.center) -- (v3.center) -- cycle;
\fill [gray!20] (v0.center) -- (v4.center) -- (v5.center) -- cycle;
\end{scope}
\draw [line width=1pt] (v0) -- (v4);
\draw [line width=1pt] (v0) -- (v5);
\draw [line width=1pt] (v4) -- (v5); 
\fill [gray!50] (v0.center) -- (v5.center) -- (v6.center) -- cycle;
\draw [line width=1pt] (v0) -- (v6);
\draw [line width=1pt] (v5) -- (v6);
\fill [gray!50] (v0.center) -- (v5.center) -- (v7.center) -- cycle;
\draw [line width=1pt] (v0) -- (v7);
\draw [line width=1pt] (v5) -- (v7);
\fill [gray!50] (v0.center) -- (v5.center) -- (v8.center) -- cycle;
\draw [line width=1pt] (v0) -- (v8);
\draw [line width=1pt] (v5) -- (v8);
\draw [fill=black] (0.17626925405727717,0.058890499875125395) circle (.5pt);
\draw [fill=black] (0.18645598603640848,-0.005140386850842488) circle (.5pt);
\draw [fill=black] (0.1733587592060968,-0.06480553130003984) circle (.5pt);
\draw [fill=black] (0.13697757356634216,-0.12156018089805683) circle (.5pt);
\draw [fill=black] (-0.18608735491467907,0.019598819384190558) circle (.5pt);
\draw [fill=black] (-0.17735587036113795,-0.05752929417208894) circle (.5pt);
\draw [fill=black] (-0.13660894244461277,-0.12883641802600773) circle (.5pt);
\draw [fill=black] (-0.06821231344187403,-0.17540433564489347) circle (.5pt);
\draw [fill=black] (0.007460552688815624,-0.1870463150496149) circle (.5pt);
\draw [fill=black] (0.08458866624509546,-0.16521760366576219) circle (.5pt);
\draw [fill=black] (-1.5,2.4) circle (.5pt);
\draw [fill=black] (-1.4,2.45) circle (.5pt);
\draw [fill=black] (-1.3,2.5) circle (.5pt);
\end{scriptsize}
\end{tikzpicture}  
\caption{Complex $K_1$.}
\label{fig:k1}
\end{subfigure}
\begin{subfigure}[h]{.28\linewidth}
\begin{tikzpicture}[line cap=round,line join=round,>=triangle 45,x=1cm,y=1cm]
\begin{scriptsize}
\begin{pgfonlayer}{foreground}
\node[dot=3pt] (v0) at (0,0) {};
\node[dot=3pt] (v1) at (0.7865072033647845,2.37305845251551) {};
\node[dot=3pt] (v1) at (0.7865072033647845,2.37305845251551) {};
\node[dot=3pt] (v2) at (1.4945672579355689,2.0064266280775427) {};
\node[dot=3pt] (v3) at (2.0485511735311537,1.4325014320868856) {};
\node[dot=3pt] (v4) at (-2.278994628732195,1.0277890074745213) {};
\node[dot=3pt] (v5) at (-1.8399580316361015,1.6924994658249684) {};
\node[dot=3pt] (v6) at (-0.8213486526906792,2.688734670322916) {};
\node[dot=3pt] (v7) at (-1.0054455436501044,2.638978753847396) {};
\node[dot=3pt] (v8) at (-1.7169551492500446,2.2309802387481317) {};
\node[dot=3pt] (v9) at (1.18,-1.9) {};
\end{pgfonlayer}
\draw [line width=1pt] (v0) -- (v1);
\draw [line width=1pt] (v0) -- (v2);
\draw [line width=1pt] (v0) -- (v3);
\draw [line width=1pt] (v2) -- (v1);
\draw [line width=1pt] (v2) -- (v3);
\begin{scope}[on background layer] 
\fill [gray!20] (v0.center) -- (v1.center) -- (v2.center) -- cycle;
\fill [gray!20] (v0.center) -- (v2.center) -- (v3.center) -- cycle;
\fill [gray!20] (v0.center) -- (v4.center) -- (v5.center) -- cycle;
\end{scope}
\draw [line width=1pt] (v0) -- (v4);
\draw [line width=1pt] (v0) -- (v5);
\draw [line width=1pt] (v4) -- (v5); 
\fill [gray!50] (v0.center) -- (v5.center) -- (v6.center) -- cycle;
\draw [line width=1pt] (v0) -- (v6);
\draw [line width=1pt] (v5) -- (v6);
\fill [gray!50] (v0.center) -- (v5.center) -- (v7.center) -- cycle;
\draw [line width=1pt] (v0) -- (v7);
\draw [line width=1pt] (v5) -- (v7);
\fill [gray!50] (v0.center) -- (v5.center) -- (v8.center) -- cycle;
\draw [line width=1pt] (v0) -- (v8);
\draw [line width=1pt] (v5) -- (v8);
\draw [fill=black] (0.17626925405727717,0.058890499875125395) circle (.5pt);
\draw [fill=black] (0.18645598603640848,-0.005140386850842488) circle (.5pt);
\draw [fill=black] (0.1733587592060968,-0.06480553130003984) circle (.5pt);
\draw [fill=black] (0.13697757356634216,-0.12156018089805683) circle (.5pt);
\draw [fill=black] (-0.18608735491467907,0.019598819384190558) circle (.5pt);
\draw [fill=black] (-0.17735587036113795,-0.05752929417208894) circle (.5pt);
\draw [fill=black] (-0.13660894244461277,-0.12883641802600773) circle (.5pt);
\draw [fill=black] (-0.06821231344187403,-0.17540433564489347) circle (.5pt);
\draw [fill=black] (0.007460552688815624,-0.1870463150496149) circle (.5pt);
\draw [fill=black] (0.08458866624509546,-0.16521760366576219) circle (.5pt);
\draw [fill=black] (-1.5,2.4) circle (.5pt);
\draw [fill=black] (-1.4,2.45) circle (.5pt);
\draw [fill=black] (-1.3,2.5) circle (.5pt);
\draw [line width=1pt] (v0) -- (v9);
\draw [line width=1pt] (v1) -- (v9);
\draw [line width=1pt] (v2) -- (v9);
\draw [line width=1pt] (v3) -- (v9);
\draw [line width=1pt] (v4) -- (v9);
\draw [line width=1pt] (v5) -- (v9);
\end{scriptsize}
\end{tikzpicture} 
\caption{Complex $K_2$.}
\label{fig_k2}
\end{subfigure}
\caption{Depiction of $K_1$ and $K_2$.}
\label{fig:k1k2}
\end{figure}

Each construction consists of two parts: the complexity study of the reduction of a specific boundary matrix, and the existence of a simplicial complex realizing said boundary matrix. 
\medskip

\textbf{Existence of simplicial complex $K_1$.}
The complex is depicted in \cref{fig:k1}. 
We start with a structure that we call an (open) wheel:
it consists of $n$ triangles incident to a \define{wheel center} vertex
such that subsequent triangles share an edge, but the first and last triangles do not share an edge. 
The edges incident with the wheel center are called \define{spoke edges}. 
The two that are part of only one triangle are called the \define{initial} and the \define{final} spoke edge, respectively.
The edges not-incident with the wheel center are called \define{tire edges}. 
We enumerate these edges as follows: first the initial spoke edge, then all the tire edges along the wheel starting with the one adjacent to the initial spoke, and finally all the remaining spoke edges following the wheel starting with the spoke forming a triangle with the initial spoke edge. 
We then sort the triangles following the wheel starting from the final spoke edge (that is, in the opposite direction w.r.t. the spoke edges).  
Note that by design, the tire edges are merging components in the filtration and are, therefore negative.

Next, we attach a \define{fan} of size $n$ to the final spoke edge. 
This means we introduce $n$ additional vertices $v_1,\ldots,v_n$ and form $n$ \define{fan triangles}, each joining one $v_i$ with the final spoke edge. 
The two edges of a fan triangle, not being the final spoke edge, are called \define{fan edges}. 
We call the fan edge incident to the center \define{center fan edge}, and the other one \define{outer fan edge}. 
We sort the edges of the filtration by letting the center fan edges come after the tire edges, followed by the outer fan edges, followed by the final spoke edge. 
Finally, the fan triangles come after the wheel triangles.

\begin{figure}[h!]
\centering
\scriptsize{
$
\begin{array}{c}
    \text{wheel triangles \ fan triangles \ \ }\\
\left(
\begin{array}{cccc|cccc}
\ast & \ast & \ast & \ast & \ast & \ast & \ast & \ast \\
\vdots & \vdots & \vdots & \vdots & \vdots & \vdots & \vdots & \vdots \\
\ast & \ast & \ast & \ast & \ast & \ast & \ast & \ast \\
0 & 0 & 0 & 1 & 0 & 0 & 0 & 0 \\
0 & 0 & 1 & 1 & 0 & 0 & 0 & 0 \\
0 & 1 & 1 & 0 & 0 & 0 & 0 & 0 \\
1 & 1 & 0 & 0 & 0 & 0 & 0 & 0 \\
1 & 0 & 0 & 0 & 1 & 1 & 1 & 1 \\
\end{array}
\right)
\end{array}
$
}
\caption{(Sub)matrix of $K_1$}
\label{fig:boundary_K1}
\end{figure}

\textbf{Reduction complexity of $K_1$.}
This construction yields a filtered simplicial complex
whose boundary matrix contains the matrix of \cref{fig:boundary_K1} as submatrix. 
Note that \cref{fig:boundary_K1} is also a submatrix of the construction in~\cite{morozov2005worst-case}.
The first half of the matrix contains a ``staircase'' of columns with decrementing pivots. 
The staircase is of size $4$ in \cref{fig:boundary_K1} but can easily be extended to an arbitrary $n$ in the obvious way. 
The second half consists of columns that all have the same pivot, equal to the lowest step of the staircase. 
When reducing the matrix (using the standard or twist algorithm), the reduction of each column in the second half causes the algorithm to add each column of the first half to it in order. 
In total, this causes a quadratic number of column operations.
The swap algorithm has the same complexity since here no swapping happens.

The retrospective algorithm, however, sparsifies the first column before it gets added to the second half. 
This simplification requires linear time (by one iteration through the staircase) and results in a unit vector. 
In all additions to the second half, the cost is therefore constant. 
This leads to linear complexity.
It is important to note that the third nonzero entry for the left-hand-side columns comes from a tire edge, which is negative. 
Hence, these indices will be removed when reducing wheel triangles. 
Therefore, ignoring these indices, we observe that indeed the first column turns into a unit vector. 
Reducing the fan triangles with the retrospective algorithm thus results in removing the final spoke edge, which makes the outer fan edge the pivot, and the reduction of the column stops after one bitflip.

It remains to be argued that the reduction of the edges in the filtration is linear as well for the retrospective reduction.
However, this is simple to see, putting an appropriate order of the vertices in the complex. 
We omit the details.
\medskip

\textbf{Existence of simplicial complex $K_2$.}
For $K_2$, we extend the complex $K_1$ by adding one more vertex, called the \define{apex}, and connecting it with every vertex of the wheel via an edge (see \cref{fig_k2}).
We call these edges \define{apex edges}. 
We sort the edges of $K_2$ in the following order: apex edges, center fan edges, initial spoke edge, tire edges, inner spoke edges, outer fan edges, and final spoke edge. 
The triangles remain in the same order as in $K_1$.

The two major differences to the situation of $K_1$ are: the tire edges are now positive edges, so the compression will not remove these entries anymore.
Moreover, shifting the outer fan edges later in the filtration creates a block of $n$ edges between the inner spoke edges and the final spoke edge.
The filtration boundary matrix therefore looks as depicted in \cref{fig:boundary_K2}, where the just mentioned block is given between the two horizontal lines.

\textbf{Reduction complexity of $K_2$.} 
We see now that twist and swap reduction only cause one column addition for every column on the right because the entries in the newly inserted block prevent the algorithm from doing further reductions. 
Importantly, each column operation
only causes a constant number of bitflips, so that the complexity is linear in the end.
Again, it can easily be argued that the reduction of the edges for the twist and swap algorithm requires only linear time.

For the retrospective reduction, the addition of the first column to the second half causes
a ``sparsification'', as in the previous example. 
However, in this case, this sparsification actually turns the first column into a column with $n$ nonzero entries because it
collects all indices of tire edges while iterating through the staircase. 
Since we then add this column $n$ times (once to every column on the right) and each addition causes $n$ bitflips, we get quadratic complexity.

\begin{figure}[h!]
\centering
\scriptsize{
$
\begin{array}{r}
\text{\ }\\
\begin{array}{r}
\\
\\
\\
\text{apex edges, ...} \\
\text{inner spoke edges}\\
\\
\\
\\
\\
\text{outer fan}\\
\text{edges}\\
\\
\text{final spoke} \\
\end{array}
\end{array}
$
\hspace{-.7cm}
$
\begin{array}{c}
    \text{wheel triangles \ fan triangles \ \ }\\
\left(
\begin{array}{cccc|cccc}
0 & 0 & 0 & 1 & 0 & 0 & 0 & 1 \\
0 & 0 & 1 & 0 & 0 & 0 & 1 & 0 \\
0 & 1 & 0 & 0 & 0 & 1 & 0 & 0 \\
1 & 0 & 0 & 0 & 1 & 0 & 0 & 0 \\
0 & 0 & 0 & 1 & 0 & 0 & 0 & 0 \\
0 & 0 & 1 & 1 & 0 & 0 & 0 & 0 \\
0 & 1 & 1 & 0 & 0 & 0 & 0 & 0 \\
1 & 1 & 0 & 0 & 0 & 0 & 0 & 0 \\
\hline
0 & 0 & 0 & 0 & 1 & 0 & 0 & 0 \\
0 & 0 & 0 & 0 & 0 & 1 & 0 & 0 \\
0 & 0 & 0 & 0 & 0 & 0 & 1 & 0 \\
0 & 0 & 0 & 0 & 0 & 0 & 0 & 1 \\
\hline
1 & 0 & 0 & 0 & 1 & 1 & 1 & 1 \\
\end{array}
\right)
\end{array}
$
}
\caption{Matrix for $K_2$ (with $n=4$)}
\label{fig:boundary_K2}
\end{figure}

\textbf{Reduction complexity of $K_3$.} 
For $K_3$, we build up a filtration boundary matrix as depicted in \cref{fig:boundary_K3}. 
For reference, we call the horizontal blocks in figure block 1 to block 5, starting from the bottom.

The twist algorithm applied to this boundary matrix reduces the fifth column by adding all columns on the left to it. 
This results in a fill-in of row indices in block 4, and the pivot being the unique row index of block 2. 
Since all columns on the right have the same pivot, 
the reduced fifth column with $n$ entries is added to every column to the right, resulting in quadratic complexity.

In the swap reduction, the fifth column is reduced in the same way.
However, when added to the first column to the right, a swap happens so that in the reduction of the subsequent columns, the sixth column is used. 
We can observe by the block structure in block 3 that
all further columns are reduced after one column addition.
Also, column six has only $3$ nonzero entries, so the total
complexity is linear. 

\textbf{Existence of simplicial complex $K_3$.}
To realize the depicted matrix as the boundary matrix of a simplicial complex, we again construct an open wheel. 
We attach one triangle to the final spoke edge, joining it with a new vertex (represented by the middle column of the matrix). 
Then, on the edge of that triangle not incident to the wheel center, we attach a fan of $n$ triangles. 
It is easily possible to sort the edges of this complex in a way that we get the depicted block structure.
\medskip

\begin{figure}[h]
\centering
\begin{minipage}{0.1\textwidth}
\scriptsize{
\begin{tabular}{cc}
    & \\
    \multirow{3}{*}{block 5} &  \\
    & \\
    & \\
    \multirow{4}{*}{block 4} &  \\
    & \\
    & \\
    & \\
    \multirow{4}{*}{block 3} &  \\
    & \\
    & \\ 
    & \\
    block 2 & \\
    \multirow{4}{*}{block 1} &  \\
    & \\
    & \\
    & 
\end{tabular}
}
\end{minipage}
\begin{minipage}{0.5\textwidth}
\scriptsize{
$
\left(
\begin{array}{cccc|c|cccc}
\ast & \ast & \ast & \ast & \ast & \ast & \ast & \ast & \ast \\
\vdots & \vdots & \vdots & \vdots & \vdots & \vdots & \vdots & \vdots & \vdots \\
\ast & \ast & \ast & \ast & \ast & \ast & \ast & \ast & \ast \\
\hline
0 & 0 & 0 & 1 & 0 & 0 & 0 & 0 & 0 \\
0 & 0 & 1 & 0 & 0 & 0 & 0 & 0 & 0 \\
0 & 1 & 0 & 0 & 0 & 0 & 0 & 0 & 0 \\
1 & 0 & 0 & 0 & 0 & 0 & 0 & 0 & 0 \\
\hline
0 & 0 & 0 & 0 & 0 & 1 & 0 & 0 & 0 \\
0 & 0 & 0 & 0 & 0 & 0 & 1 & 0 & 0 \\
0 & 0 & 0 & 0 & 0 & 0 & 0 & 1 & 0 \\
0 & 0 & 0 & 0 & 0 & 0 & 0 & 0 & 1 \\
\hline
0 & 0 & 0 & 0 & 1 & 1 & 1 & 1 & 1 \\
\hline
0 & 0 & 1 & 1 & 0 & 0 & 0 & 0 & 0 \\
0 & 1 & 1 & 0 & 0 & 0 & 0 & 0 & 0 \\
1 & 1 & 0 & 0 & 0 & 0 & 0 & 0 & 0 \\
1 & 0 & 0 & 0 & 1 & 0 & 0 & 0 & 0 \\
\end{array}
\right)
$
}
\end{minipage}
\caption{Matrix for $K_3$ (with $n=4$)}
\label{fig:boundary_K3}
\end{figure}

\textbf{Reduction complexity of $K_4$.}
For $K_4$, we build a boundary matrix as in \cref{fig:boundary_K4}.
For general $n$, the matrix has $n$ columns on the left (group A), $n$ columns on the right (group C), and exactly $3$ columns in the middle (group B).

For the twist reduction, the 2-nd column in group B gets reduced with one column addition, resulting in a column with $4$ nonzero entries.
The 3-rd column in group B then has the same pivot as the just-reduced column in the middle, and the reduction of the 3-rd column requires the addition of all columns in group A. 
Still, this process only requires a linear amount of bitflips.
All columns in group C get added from the 2-nd column in group B, and because of their entries in the 3-rd row block, the reduction stops after one addition. 
In total, the twist reduction needs only linear time.

In the swap reduction, the difference is that at the beginning of the reduction of the 3-rd column in group B, a swap happens (as the 3rd column has only $3$ nonzero entries, the 2-nd column has $4$ entries).
That means that the 3-rd column in the middle gets added to all columns in group C. 
Consequently, for the reduction of every column on the right, the reduction adds all the groups on the left to it, resulting in a quadratic number of column additions.

\textbf{Existence of simplicial complex $K_4$.}
The construction of a complex $K_4$ that realizes this boundary matrix can be done as follows:
Similarly to $K_3$, we start with an open wheel, attach one new triangle (that is the 3-rd column in group B), and put a fan of $n$ triangles at its outer edge (these are the columns forming group C).
To one of these triangles (that is, the 1-st column in group B) we attach another triangle (the 2-nd column in group B).
The edge that is shared among the last described triangles is the last row in the matrix.
The edges can easily be sorted to yield the matrix of \cref{fig:boundary_K4}.

\begin{figure}[h]
\centering
\scriptsize{
$
\begin{array}{c}
    \text{group A \qquad group B \qquad group C }\\
\left(
\begin{array}{cccc|ccc|cccc}
\ast & \ast & \ast & \ast & \ast & \ast & \ast & \ast & \ast & \ast & \ast \\
\vdots & \vdots & \vdots & \vdots & \vdots & \vdots & \vdots & \vdots & \vdots & \vdots & \vdots \\
\ast & \ast & \ast & \ast & \ast & \ast & \ast & \ast & \ast & \ast & \ast \\
\hline
0 & 0 & 0 & 0 & 0 & 0 & 1 & 0 & 0 & 0 & 0 \\
0 & 0 & 0 & 0 & 0 & 1 & 0 & 0 & 0 & 0 & 0 \\
0 & 0 & 0 & 0 & 0 & 1 & 0 & 0 & 0 & 0 & 0 \\
0 & 0 & 0 & 0 & 1 & 0 & 0 & 0 & 0 & 0 & 0 \\
\hline
0 & 0 & 0 & 0 & 0 & 0 & 0 & 1 & 0 & 0 & 0 \\
0 & 0 & 0 & 0 & 0 & 0 & 0 & 0 & 1 & 0 & 0 \\
0 & 0 & 0 & 0 & 0 & 0 & 0 & 0 & 0 & 1 & 0 \\
0 & 0 & 0 & 0 & 0 & 0 & 0 & 0 & 0 & 0 & 1 \\
\hline
0 & 0 & 1 & 1 & 0 & 0 & 0 & 0 & 0 & 0 & 0 \\
0 & 1 & 1 & 0 & 0 & 0 & 0 & 0 & 0 & 0 & 0 \\
1 & 1 & 0 & 0 & 0 & 0 & 0 & 0 & 0 & 0 & 0 \\
1 & 0 & 0 & 0 & 0 & 0 & 1 & 0 & 0 & 0 & 0 \\
\hline
0 & 0 & 0 & 0 & 1 & 0 & 1 & 1 & 1 & 1 & 1 \\
0 & 0 & 0 & 0 & 1 & 1 & 0 & 0 & 0 & 0 & 0 \\
\end{array}
\right)
\end{array}
$
}
\caption{Matrix for $K_4$ (with $n=4$)}
\label{fig:boundary_K4}
\end{figure}

\section{Conclusion and Discussion.}
In this work, we analyzed how the sparsity of the reduced matrix correlates with the efficiency of the reduction by comparing different algorithms that keep the matrix sparse(r). 
The experiments show that there is no direct relation, as algorithms resulting in less sparse matrices were faster than others that aggressively sparsify. 
Nevertheless, the idea of keeping the matrix sparse
has led us to novel reduction strategies that improve upon state-of-the-art reductions.
Hence, sparsity is an important factor in fast matrix reduction.

The retrospective algorithm often achieves comparable or even better performance than the twist reduction without clearing columns. 
Specifically, it outperforms all other tested methods for shuffled filtration. 
Up to our knowledge, this is the first time that a method without clearing has been proven competitive in practice, which is remarkable as the clearing is the standard optimization that consistently leads to improved performances.
In our experiments over a wide range of datasets, the retrospective method has regularly low fill-in compared to the other methods. 
We believe that the superior performance of the retrospective method is rooted in its sparsity-preserving property.

As indicated in \cref{sec:differentiating}, there is
no strategy that is strictly better than others,
so the best choice of reduction for a specific type
of input has to be determined by comparison.
We will integrate our novel variants into the PHAT library in the next release of PHAT
to facilitate this comparison.

\end{document}